\documentclass[a4paper,twocolumn,notitlepage,nofootinbib,longbibliography,superscriptaddress,floatfix]{revtex4-2}
\usepackage{
    adjustbox,
    algorithm,
    algpseudocode,
    amsmath,
    amssymb,
    amsthm,
    booktabs,
    braket,
    comment,
    csquotes,
    enumitem,
    graphicx,
    IEEEtrantools,
    mathtools,
    multirow,
    natbib,
    physics,
    refcount,
    relsize,
    url,
    times,
    xcolor
}

\usepackage[english]{babel}
\usepackage[normalem]{ulem}
\usepackage[caption=false]{subfig}

\definecolor{mainblue}{HTML}{1f77b4}
\definecolor{mainorange}{HTML}{ff7f0e}
\definecolor{maingreen}{HTML}{2ca02c}
\definecolor{mainred}{HTML}{DC3522}
\definecolor{mainpurple}{HTML}{9467bd}
\definecolor{mainpink}{HTML}{e377c2}

\usepackage[colorlinks=true, urlcolor=mainorange, linkcolor=maingreen, citecolor=mainblue]{hyperref}

\DeclareMathOperator{\Herm}{Herm}

\renewcommand{\exp}{\ensuremath{\mathrm{exp}}}

\providecommand{\to}{\ensuremath{\Tilde{o}}}

\providecommand{\tt}{\ensuremath{\Tilde{t}}}

\providecommand{\calD}{\ensuremath{\mathcal{D}}}

\providecommand{\calG}{\ensuremath{\mathcal{G}}}
\providecommand{\calH}{\ensuremath{\mathcal{H}}}

\providecommand{\calM}{\ensuremath{\mathcal{M}}}

\providecommand{\calO}{\ensuremath{\mathcal{O}}}

\DeclareMathOperator*{\bbE}{\ensuremath{\mathbb{E}}}

\providecommand{\bbI}{\ensuremath{\mathbb{I}}}

\providecommand{\bbR}{\ensuremath{\mathbb{R}}}

\providecommand{\bbU}{\ensuremath{\mathbb{U}}}

\DeclareMathOperator*{\argmin}{arg\,min}

\newtheorem{theorem}{Theorem}
\newtheorem{lemma}[theorem]{Lemma}

\newtheorem{corollary}[theorem]{Corollary}

\newtheorem{innercustomgeneric}{\customgenericname}
\providecommand{\customgenericname}{}
\newcommand{\newcustomtheorem}[2]{%
  \newenvironment{#1}[1]
  {%
   \renewcommand\customgenericname{#2}%
   \renewcommand\theinnercustomgeneric{##1}%
   \innercustomgeneric
  }
  {\endinnercustomgeneric}
}

\newcustomtheorem{customtheorem}{Theorem}
\newcustomtheorem{customlemma}{Lemma}
\newcustomtheorem{customproposition}{Proposition}

\newcommand{\fu}{Dahlem Center for Complex Quantum Systems, Freie Universit\"{a}t Berlin, 14195 Berlin, Germany}

\newcommand{\hhi}{Fraunhofer Heinrich Hertz Institute, 10587 Berlin, Germany}
\newcommand{\leiden}{Leiden University, Niels Bohrweg 1, 2333 CA Leiden, Netherlands}
\newcommand{\vw}{Volkswagen Group Innovation, Berliner Ring 2, 38440 Wolfsburg, Germany}
\newcommand{\physwaterloo}{Department of Physics and Astronomy, University of Waterloo, ON N2L 3G1, Canada}
\newcommand{\vinst}{Vector Institute, Toronto, ON M5G 0C6, Canada}

\begin{document}

\title{Double descent in quantum kernel methods}

\author{Marie Kempkes}
\email{marie.kempkes@volkswagen.de}
\affiliation{\leiden}
\affiliation{\vw}
\author{Aroosa Ijaz}
\email{a4ijaz@uwaterloo.ca}
\affiliation{\physwaterloo}
\affiliation{\vinst}
\affiliation{\fu}
\author{Elies Gil-Fuster}
\affiliation{\fu}
\affiliation{\hhi}
\author{Carlos Bravo-Prieto}
\affiliation{\fu}
\author{Jakob Spiegelberg}
\affiliation{\vw}
\author{Evert van Nieuwenburg}
\affiliation{\leiden}
\author{Vedran Dunjko}
\affiliation{\leiden}

\begin{abstract}
    The double descent phenomenon challenges traditional statistical learning theory by revealing scenarios where larger models do not necessarily lead to reduced performance on unseen data. While this counterintuitive behavior has been observed in a variety of classical machine learning models, particularly modern neural network architectures, it remains elusive within the context of quantum machine learning. In this work, we analytically demonstrate that linear regression models in quantum feature spaces can exhibit double descent behavior by drawing on insights from classical linear regression and random matrix theory. Additionally, our numerical experiments on quantum kernel methods across different real-world datasets and system sizes further confirm the existence of a test error peak, a characteristic feature of double descent. Our findings provide evidence that quantum models can operate in the modern, overparameterized regime without experiencing overfitting, potentially opening pathways to improved learning performance beyond traditional statistical learning theory.
\end{abstract}

\maketitle

\section{Introduction} \label{s:introduction}

    \begin{figure*}
        \centering
        \includegraphics{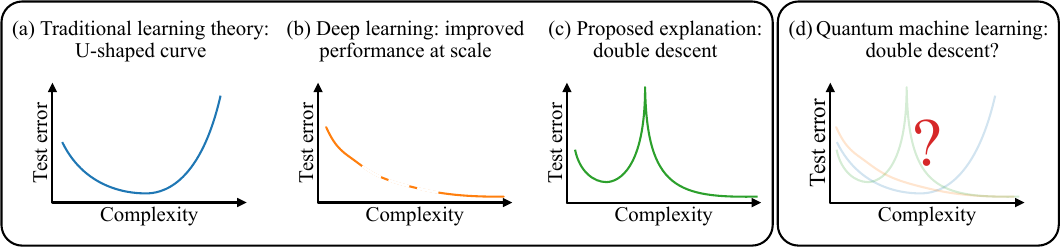}
        \caption{
            \textbf{Conceptual overview of this work.}
            (a) Traditional statistical learning theory predicts a U-shaped relationship between model complexity and test error, where both overly simple and overly complex models have poor generalization.
            (b) Empirical observations in deep learning challenge this view, showing that larger models often exhibit improved performance at scale, although the precise relationship between complexity and error remains elusive.
            (c) The double descent phenomenon provides a theoretical and empirical explanation for this behavior, highlighting a non-monotonic relationship between model complexity and test error.
            (d) This work investigates whether a similar double descent behavior can emerge in quantum machine learning models, offering new insights into their generalization properties.
            The x-axis represents a wide notion of complexity, including for example the size of the training set.
        } 
        \label{fig:4step}
    \end{figure*}

    Recent progress in machine learning (ML) has significantly enhanced data-driven insights across science and industry.
    At the same time, there is growing interest in understanding how quantum computers can offer improvements in solving computational problems in machine learning.
    This question is central within the field of quantum machine learning (QML), which seeks to harness quantum effects for learning from data~\cite{biamonte2017quantum, dunjko2018machine, carleo2019machine, schuld2021machine, cerezo2021variational, bharti2022noisy}.
    Researchers have identified specific learning problems for which QML could potentially outperform their classical counterparts~\cite{liu2021rigorous, sweke2021quantum, huang2021information, gyurik2023exponential, pirnay2023superpolynomial, molteni2024exponential}.
    However, while similarities and differences between QML and classical ML are continually being elucidated~\cite{gil2024relation, schuld2021supervised, landman2022classically, gil2024understanding, sweke2023potential}, a critical distinction remains: classical ML has demonstrated large-scale success in real-world applications, whereas QML has yet to achieve comparable practical breakthroughs.

    Understanding this gap requires deeper insight into the generalization properties of the models, i.e., their performance on unseen data.
    Traditional statistical learning theory suggests that increasing the so-called model complexity, typically quantified by measures such as the number of trainable parameters, the VC dimension~\cite{vapnik1971uniform, valiant1972learnable,Vapnik1999SLT}, among others~\cite{ShalevShwartz2014understanding},
    often leads to poor generalization.
    In this context, complexity can include not only model-specific characteristics but also broader factors, such as the size of the training dataset.
    The behavior predicted by traditional learning theory is visualized as a U-shaped curve, where overly complex models tend to memorize training data rather than capturing meaningful patterns (Fig.~\ref{fig:4step}(a)).
    However, empirical evidence from deep learning over the past decade challenges this notion.
    Large neural networks often exhibit improved generalization at scale~\cite{krizhevsky2012advances, neyshabur2015insearch, zhang2017understanding, novak2018sensitivity, nakkiran2021deep, belkin2019reconciling}, contradicting the theoretically predicted trade-off between complexity and generalization (Fig.~\ref{fig:4step}(b)).
    A growing body of literature has sought to explain this surprising behavior~\cite{cooper2018losslandscape, spigler2019jamming, neyshabur2018theroleof, Zhu2019aconvergencetheory, ADVANI2020High, Adlam2020understanding, geiger2020scaling, dar2021farewellbiasvariancetradeoff, Jason2022Memorizing, curth2023Uturn, schaeffer2023doubledescentdemystifiedidentifying}, attributing it in part to the so-called \textit{double descent} phenomenon.
    
    Double descent describes a non-monotonic relationship between test error and model complexity.
    Initially, the test error follows the expected U-shaped curve but then, counterintuitively, begins to decrease again as model complexity surpasses a critical threshold (Fig.~\ref{fig:4step}(c)).
    This behavior marks the transition from the underparameterized regime, where the model lacks sufficient capacity to represent the data, to the overparameterized regime, where the model has more parameters than necessary to fit the training dataset.
    Beyond this threshold, as observed empirically, the test error often decreases asymptotically until it saturates.
    It is important to note, however, that there is no guarantee of achieving saturation below the minimum of the U-shaped curve.
    These insights naturally motivate the research question central to our work (Fig.~\ref{fig:4step}(d)): \emph{Are there quantum machine learning models that exhibit double descent behavior?} 

    In this work, we provide a theoretical analysis of double descent for linear regression models in quantum feature spaces, prominently including quantum kernel methods.
    By leveraging the fact that those models act linearly in their corresponding Hilbert space, we derive an expression for the test error, building on insights from Ref.~\cite{schaeffer2023doubledescentdemystifiedidentifying}.
    Using the Mar\v{c}enko-Pastur law, a seminal result in random matrix theory, we analytically demonstrate that, in the asymptotic limit and under reasonable assumptions about the data distribution, the test error peaks at the interpolation threshold as the quantum model transitions from underparameterization to overparameterization, a key feature of the double descent phenomenon.
    We further test our theoretical findings with numerical experiments on quantum kernels, conducted on several datasets, including real-world data, for varying system sizes and different families of quantum feature maps.
    Consequently, our results confirm that \emph{certain quantum models exhibit the characteristic double descent behavior}.

    To the best of our knowledge, this work represents the first empirical and theoretical observation of double descent in quantum models.
    Our findings provide a foundational perspective on this phenomenon, paving the way for further exploration into its implications for model performance and generalization in quantum machine learning.

\section{Related concepts in QML} \label{s:relatedworks}

    The double descent phenomenon challenges the validity of conclusions drawn from statistical learning theory, thereby undermining the relevance of generalization bounds derived from such framework.
    In a similar vein, Ref.~\cite{gil2024understanding} recently revealed flaws in the prevalent approach to studying generalization in QML, known as \emph{uniform} generalization bounds.
    These bounds are uniform over the entire set of functions the learning model can output, regardless of the data distribution, the learning algorithm, or the specific function selected by the model.
    Earlier, Ref.~\cite{peters2023generalization} took a significant step away from uniform generalization bounds by showcasing instances of \emph{benign overfitting} in QML.
    In particular, the authors identified scenarios where quantum learning models could fit the training data perfectly while still performing well on new test data, thus defying the traditional picture of statistical learning theory.

    A well-known obstacle to successful QML at scale has been identified as exponential concentration of the functions output by the learning model~\cite{larocca2024review}.
    While exponential concentration is often linked to problems in trainability~\cite{gil2024relation}, Ref.~\cite{thanasilp2024exponentialconcentrationquantumkernel} revealed that it can also adversely affect generalization.
    Specifically, exponential concentration may manifest itself in quantum feature maps suffering from a \emph{vanishing similarity} problem~\cite{thanasilp2024exponentialconcentrationquantumkernel}, where the similarity measure between data points becomes exponentially small as the number of qubits increases.
    As a result, the number of measurement shots required to reliably evaluated the labeling function scales exponentially, posing a significant hurdle for practical implementations.
    %Our theoretical analysis in this work is conducted under the implicit assumption that we operate with quantum kernels that do not suffer from such concentration issue. This allows us to focus on the intrinsic relationship between model complexity, data, and generalization error that gives rise to double descent, separate from the challenges posed by exponential concentration. %A detailed discussion of the interplay between double descent and the concentration problem is provided in Section~\ref{s:discussion}.
    Our theoretical analysis diverges from this particular challenge.
    We focus on the exact evaluation of the functions, avoiding reliance on finite-shot approximations.
    This distinction places our approach on a different footing than the analysis presented in Ref.~\cite{thanasilp2024exponentialconcentrationquantumkernel}.  
    %For further discussion on the concentration problem, refer to Section~\ref{s:discussion}.
    
    Lastly, let us note that our work employs the standard notion of overparameterization from classical machine learning literature.
    This choice contrasts with the alternative definition presented in Ref.~\cite{larocca2023theory}, which follows a different framework.
    For further details on related and prior work, we refer the reader to Appendix \ref{a:related_work}.

\section{Linear regression in quantum feature space} \label{s:preliminaries}

    In this section, we build upon the derivation presented in Ref.~\cite{schaeffer2023doubledescentdemystifiedidentifying} to analyze the double descent phenomenon in QML.
    Specifically, we apply their framework to matrix-valued feature maps commonly encountered in QML.
    This allows us to explore whether and how the double descent behavior manifests in QML models.
    
    In QML, classical data $x$ is mapped into a quantum feature space via a quantum feature map, denoted as $\rho(x)$.
    This mapping is achieved via a unitary transformation:
        \begin{align}
            x \mapsto \rho(x) = S(x)\rho_0 S^\dagger(x),
        \end{align}
    where $S(x)$ is a unitary operator parameterized by the data, and $\rho_0$ represents an initial $n$-qubit quantum state.
    A common approach to constructing QML models involves considering linear functions of the quantum feature state \cite{schuld2021supervised, jerbi2023quantum}.
    In particular, the model function is given by:
    \begin{align}
        f(x) &= \Tr{\rho(x)\calM},
    \end{align}
    where $\calM\in\Herm(2^n)$ is a quantum observable, which defines a hyperplane in quantum feature space.
    The function $f$ can be evaluated as the expectation value of a quantum measurement specified by $\calM$.
    An example of such quantum linear models are Quantum Kernel Methods~(QKMs), in which a kernel function that quantifies the similarity between quantum feature states plays a central role in defining the model.
    The inner product between two quantum feature states is referred to as an Embedding Quantum Kernel (EQK) function~\cite{schuld2019quantum, havlivcek2019supervised, hubregtsen2022training, gil2024expressivity}:
        \begin{align}
            \kappa_\rho(x, x') = \Tr \{\rho(x)\rho(x')\}.
            \label{eq:eqk}
        \end{align}
    EQKs can be used to find favorable observables $\calM$ tailored to specific tasks.

    In the context of supervised learning, we consider labeled data points $(x,y)\in\bbR^d\times\bbR$, where $x$ represents the input and $y$ the corresponding label.
    Given a training set $\{(x_i,y_i)\}_{i=1}^N$ of size $N$, our objective is to infer the underlying relationship between inputs and labels.
    We consider quantum linear models specified by a quantum feature map $\rho$ using $n$ qubits, and denote by $p$ the number of free parameters in the model.
    In quantum kernel methods, the number of parameters is upper bounded by the dimension of the orthonormal Hermitian basis spanned by $\rho$. Consequently, for embedding quantum kernels on $n$ qubits, the number of parameters is at most $p=4^n$. This remains valid for pure feature maps, provided that the image of $\rho$ spans the full space of density matrices. For more details on the parameter count in quantum kernel methods we refer the reader to Appendix~\ref{a:parametercount}. Note that we use the term \emph{parameters} in a fundamentally different context as in, e.g., parameterized quantum circuits, since the number of parameters here is fully specified by the data encoding map.
    %The number of parameters is therefore at most $p=4^n$, but this may vary depending on constraints.
    In quantum feature space, the training data consists of quantum density matrices and labels $\{(\rho_i, y_i)\}_{i=1}^N$.
    We omit the explicit $x$-dependence for simplicity and in order to make the framework applicable to linear regression with quantum states which do not arise from classical data.
    We introduce a data matrix $D$ and label matrix $Y$ as
    \begin{align}\label{eq:datamatrix}
        D = \begin{pmatrix} \left[\rho_1\right]^\dagger \\ \vdots \\ \left[\rho_N\right]^\dagger \end{pmatrix}, \quad\quad Y = \begin{pmatrix} y_1 \\ \vdots \\ y_N\end{pmatrix}.
    \end{align}
    We use square bracket notation to highlight that $D$ is \emph{not} a matrix with each $\rho_i^\dagger$ as a submatrix.
    Instead, as detailed in Appendix~\ref{a:regression_detail}, $D$ should be understood as a \emph{vector of co-matrices}, where $\rho_i^\dagger$ represents the $i^\text{th}$ element.
    Since Hermitian matrices form a vector space, a co-matrix belongs to the dual space of this vector space and acts as a linear map from matrices to real numbers.
    This distinction becomes crucial later when encountering expressions such as $D^\dagger D$ or $DD^\dagger$.
    Indeed, $D^\dagger = \left([\rho_1]\ldots[\rho_N]\right)$ is a \emph{co-vector of matrices}.
    With this notation, a \enquote{co-matrix acting on a matrix} is an inner product, $[\rho]^\dagger_i\rho^{\vphantom{\dagger}}_j=\Tr\{\rho_i\rho_j\}$, and a \enquote{matrix acting on a co-matrix} is a tensor product (as the canonical outer product in a vector space, in this case of Hermitian matrices), $\rho^{\vphantom{\dagger}}_i[\rho_j]^\dagger = \rho_i\otimes\rho_j$.
    Finally, the label vector $Y\in\bbR^N$ is a real-valued vector.
    
    In linear least-squares regression, the observable $\calM$, sometimes referred to as \emph{parameter matrix}, is determined by solving an optimization problem based on the given training data.
    Two regimes arise based on the relationship between the number of parameters $p$ and the number of training samples $N$.
    
    \paragraph{Underparameterized regime $(N > p)$:}
    the number of data points exceeds the number of parameters.
    The observable $\calM^u$ is obtained by solving a linear least squares problem~\cite{engl1996regularization}:
    \begin{align}\label{eq:LSsolution}
        \calM^u &= \argmin_\calM \lVert D\calM - Y\rVert^2_2 = \left(D^\dagger D\right)^{-1}D^\dagger Y\,,
    \end{align}
        where we refer to the $p\times p$-dimensional $D^\dagger D = \sum^N_{i=1} \rho_i \otimes \rho_i$ as the sample covariance matrix (we slightly abuse notation by omitting centering and normalization).
        The matrix inverse is well-defined because $D^\dagger D$ is full rank in this regime.
        
    \paragraph{Overparameterized regime $(N < p)$:} the number of parameters exceeds the number of data points.
    In this case, there is a continuum of linear functions that perfectly fit the data $M\coloneqq\{\calM\in\Herm(2^n)\,|\,\Tr\{\rho(x_i)\calM\}=y_i,i\in[N]\}$.
    Among all these, the minimum-norm solution is selected~\cite{engl1996regularization}:
    \begin{align}\label{eq:MNLSsolution}
            \calM^o &= \argmin_{\calM\in M} \lVert\calM\rVert^2_2 = D^\dagger\left(DD^\dagger\right)^{-1} Y\,,
        \end{align}
        where $DD^\dagger = \left(\Tr{\rho_i \rho_j}\right)^N_{i,j=1}$ is the $N \times N$ Gram matrix, and its inverse is well-defined in this regime because $DD^\dagger$ is full rank.
    The condition $N=p$ is referred to as the \emph{interpolation threshold}, marking the transition where a model starts to interpolate (reaching nearly zero training error) and shifts from underparameterization to overparameterization.
    
    We compare the models obtained in the two regimes to a hypothetical optimal linear model with corresponding $\calM^\ast$, which achieves the best possible performance:
    \begin{align}
        \calM^\ast &= \argmin_{\calM}\left\{\bbE_{(\rho,y)} \left[ \left(\Tr{\rho\calM} - y\right)^2 \right] \right\},
    \end{align}
    where the expectation value is taken over the underlying data distribution, which we assume to be unknown.
    The true relation between inputs and outputs, called \emph{ground truth}, need not be a deterministic linear function in feature space.
    That means $\calM^\ast$ may still incur some error on individual samples, either due to noisy labels or because the best linear model in feature space cannot accurately capture the ground truth function, a situation known as model misspecification~\cite{dar2021farewellbiasvariancetradeoff}. Note that in this context, minimizing the expected risk therefore does not guarantee good performance.
    
    Consequently, when presented with a new test point $(\rho_t,y_t)$ not present in the training set, neither the empirical risk minimizers $\calM^{u,o}$ nor the expected risk minimizer $\calM^\ast$ are guaranteed to predict the correct label, i.e., $y^{(u,o,\ast)}_t \neq y_t$.
    Here, $y^{(u,o,\ast)}_t \coloneqq \Tr{\rho_t\calM^{(u,o,\ast)}}$ denotes the predicted label on input $\rho_t$ for the underparameterized empirical risk minimizer $(u)$, the overparameterized empirical risk minimizer $(o)$, and the optimal linear model $(\ast)$.
    Since the expected risk minimizer $\calM^\ast$ is independent of the training dataset $\{ (\rho_i, y_i) \}_{i=1}^N$, it can generally incur prediction errors on any training datapoint  $(\rho_i,y_i)$.
    Let us define this \emph{residual} error as $e_i \coloneqq y_i - y^\ast_i$, and denote the vector of residual errors by $E=(e_i)_{i=1}^N$.
    
    We compare the predictions made by the empirical risk minimizers to those of the optimal linear model on unseen data:
    \begin{align}
        y^u_t - y^\ast_t &= \Tr{\rho_t (D^\dagger D)^{-1}D^\dagger E},\label{eq:yut-ystar} \\
        y^o_t - y^\ast_t &= \Tr{\rho_t D^\dagger(DD^\dagger)^{-1}E} \label{eq:yot-ystar} \\
        &\hphantom{=}+ \Tr{\rho_t\left(D^\dagger(DD^\dagger)^{-1}D - \bbI_{2^n\times2^n}\right)\calM^\ast}.
    \end{align}
    A detailed step-by-step derivation of these expressions can be found in Appendix~\ref{a:regression_detail}.

    Interestingly, the distinction between the two expressions (the inverse of the sample covariance matrix or of the Gram matrix) vanishes when we apply the singular value decomposition (SVD) of the data matrix $D=U\Sigma V^\dagger$.
    Here, $U\in\bbU(N)$ is a unitary matrix of left singular vectors, $\Sigma\in\bbR^{N\times(2^n\times2^n)}$ is a rectangular diagonal matrix of singular values, and $V\in\bbU(2^n\times2^n)$ is a unitary matrix of right singular vectors.
    Leveraging the properties of SVD, Eqs.~\eqref{eq:yut-ystar} and~\eqref{eq:yot-ystar} can be rewritten as:
    \begin{align}
        y^u_t-y^\ast_t &= \Tr{\rho_t V \Sigma^+ U^\dagger E}, \\
        y^o_t - y^\ast_t &= \Tr{\rho_t V \Sigma^+ U^\dagger E} \\
        &\hphantom{=} + \Tr{\rho_t\left(D^\dagger(DD^\dagger)^{-1}D - \bbI_{2^n\times2^n}\right)\calM^\ast}.\label{eq:overp_SVD}
    \end{align}
    Here, $\Sigma^+$ denotes the pseudoinverse of $\Sigma$.
    Notably, both the underparameterized and overparameterized cases share a common term, $\Tr{\rho_t V \Sigma^+ U^\dagger E}$.
    The second term in Eq.~\eqref{eq:overp_SVD} is unique to the overparameterized regime and accounts for the underdetermination of the linear optimization problem.
    
    Note that the shared term corresponds to a variance-like and the second term to a bias-like contribution to the test error. Importantly, we compare the learned solution to the optimal linear model for the given training data, rather than to the ground truth. As a result, these contributions differ from the classical definitions of bias and variance in statistical learning theory.
    We refer the reader to Appendix~\ref{a:regression_detail} for more details.
    
    Several classical works have shown that the variance-like term is responsible for the double descent behavior~\cite{Bartlett2020benign, ADVANI2020High, Adlam2020understanding, Hastie2022surprises, Jason2022Memorizing, schaeffer2023doubledescentdemystifiedidentifying, dar2021farewellbiasvariancetradeoff}. All these works consistently demonstrate that the model's sensitivity to fluctuations in the training data is maximal at interpolation, causing the variance-like term to dominate the test error. We further show in our ablation study in Fig.~\ref{fig:ablation_results}(c) that the bias-like term to the test error remains small for quantum kernels, and hence does not explain the peak.
    Our subsequent analysis of the double descent behavior thus focuses exclusively on the variance-like contribution.
    
    %This term is further analyzed in Appendix~\ref{a:regression_detail}, while our immediate focus lies on the shared term, as it has been shown to account for the double descent behavior in several classical works, e.g., Refs.~\cite{Bartlett2020benign, ADVANI2020High, Adlam2020understanding, Hastie2022surprises, Jason2022Memorizing, schaeffer2023doubledescentdemystifiedidentifying, dar2021farewellbiasvariancetradeoff}. 
    To unpack the term further, we first expand it in the orthonormal basis of singular vectors:
    \begin{align}
        \Tr{\rho_t V \Sigma^+ U^\dagger E} &= \sum_{r=1}^R \frac{1}{\sigma_r} \Tr{\rho^V_r\rho^{\vphantom{V}}_t}\langle u_r , E\rangle\,,\label{eq:dd_error}
    \end{align}
    where $R=\min\{N,2^n\times2^n\}$ denotes the rank of $D$, and the right singular vectors $\rho^V_r$ are quantum density matrices.
    This formula highlights three distinct factors that influence the prediction error:
    \begin{enumerate}
        \item The reciprocals of the singular values $1/\sigma_r$.
        \item The interaction of the test input $\rho_t$ with the basis of right singular vectors $\Tr{\rho^V_r\rho^{\vphantom{V}}_t}$.
        \item The projection of $E$ onto the left singular vectors $\langle u_r, E\rangle$.
    \end{enumerate}
    The error decomposition indicates that for double descent to occur, the three contributing factors must become large enough to cause a spike in the prediction error on new test data.

\section{Spectral analysis with random matrix theory}
\label{s:spectralanalysis}

\begin{figure*}[t!]
        \centering
        \includegraphics[width=\textwidth]{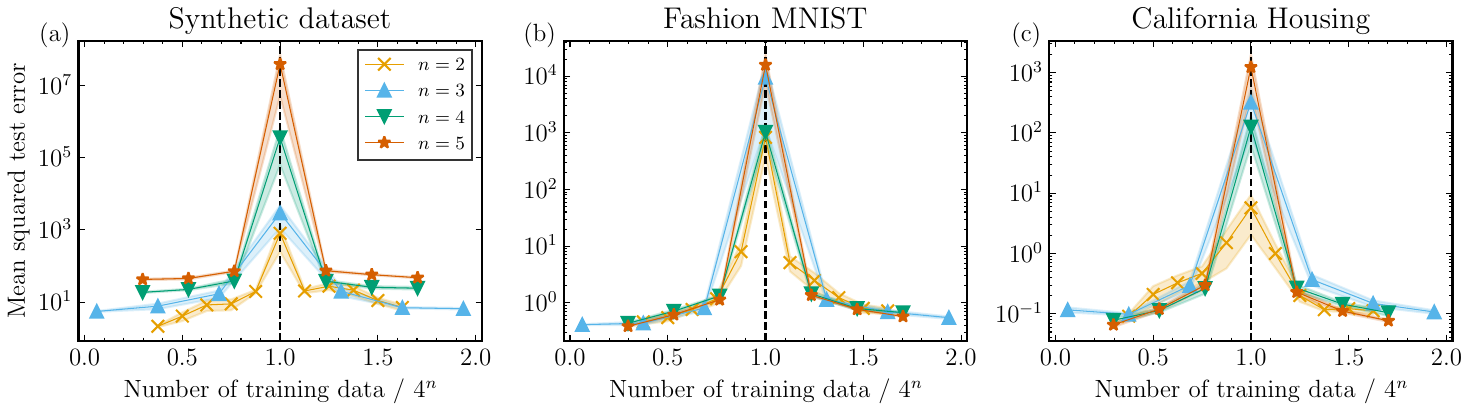}
        \caption{
            \textbf{Empirical evidence of double descent in QKMs.}
            Mean squared test error as a function of the normalized number of training data points $(N/4^n)$ for the (a) Synthetic, (b) Fashion MNIST, and (c) California Housing datasets, employing the EQK in Eq.~\eqref{eq:eqk}.
            The shaded area corresponds to the standard deviation for five independent experiment repetitions, each using independently sampled training data.
            The dotted black line indicates the interpolation threshold. The region to the left of the threshold corresponds to the overparameterized regime, while the region to the right corresponds to the underparameterized regime.
            The consistent peak across all datasets and system sizes $n$ confirms the presence of double descent.
        }
        \label{fig:dd_qkernel}
    \end{figure*}
    
    In this section, we present analytical results that predict a peak in the prediction error at the interpolation threshold $N=p$, for the problem of linear regression in quantum feature space introduced in Section~\ref{s:preliminaries}.
    Our analysis builds on the error decomposition derived in the previous section, which explicitly identified three factors contributing to the prediction error in quantum linear models, as specified in Eq.~\eqref{eq:dd_error}.
    Among the three factors, we now focus specifically on the reciprocal singular values of the data matrix $D$. Under certain conditions, $D$ behaves as a random matrix whose singular values follow a Mar\v{c}enko-Pastur (MP) law. This enables predictions about when the test error is likely to peak.
    Below, we present a simplified version of the relevant theorem and proof.
    Full proof details can be found in Appendix~\ref{a:proofthmtestpeak}.
        
    \begin{theorem}[Test error peak at interpolation -- informal]\label{thm:testpeak}
        Consider a Lipschitz continuous quantum feature map $\rho$ with $p$ linearly independent dimensions, and consider a linear regression problem in the corresponding feature space.
        Let $\{(\rho(x_i),y_i)\}_{i=1}^N$ be a training set, where $x_i \sim \mathcal{N}(0, \mathbb{I}_d)$ are $d$-dimensional i.i.d. Gaussian normal random samples.
        Then, the test error of the linear model based on $\rho(x)$ and Eqs.~\eqref{eq:LSsolution} and~\eqref{eq:MNLSsolution} peaks with high probability at $N=p$, for large enough $N$ and $p$.

        Furthermore, for a quantum feature map on $n$ qubits, the training set size $N$ at which the peak in test error occurs fulfills $N\in\calO(\exp(n))$ with high probability.
    \end{theorem}
    \begin{proof}[Proof sketch]
        The proof relies on the error decomposition introduced in Section~\ref{s:preliminaries} and results from random matrix theory. Namely, we employ the generalized Mar\v{c}enko-Pastur law, as derived in Ref.~\cite{louart2018concentration}.
        This law implies that for a Lipschitz-continuous feature map applied to Gaussian normal random vectors, the spectral behavior of the sample covariance matrix is the same as that of the sample covariance matrix constructed from a data matrix of Gaussian random entries.
        In particular, the MP law characterizes the range and probability density of eigenvalues of the sample covariance matrix as a function of system parameters $N$ and $p$.
        As $N,p\to\infty$ at some constant rate, the MP law predicts that the number of non-zero eigenvalues is minimized when $N=p$.
        Since the eigenvalues of the sample covariance matrix correspond to the squared singular values of the data matrix, Eq.~\eqref{eq:dd_error} reaches its maximum when these singular values $\sigma_r$ are smallest.
        Consequently, the test error is maximized with high probability at $N=p$, as $N,p\to\infty$.

        Additionally, the number of linearly independent dimensions of any $n$-qubit quantum feature map cannot exceed the linear dimension of the set of $2^n$-dimensional Hermitian matrices, which is $4^n$.
    \end{proof}
    
    Theorem~\ref{thm:testpeak} suggests that in the asymptotic limit, one might observe the phenomenon of double descent in quantum kernel methods if the data distribution is suitably aligned.
    Two key questions remain to ensure these results are broadly applicable: (1) are typical quantum feature maps indeed Lipschitz continuous? and (2) does the assumption of i.i.d. Gaussian inputs reasonably model realistic data sources?
    First, we establish that a general class of commonly used quantum feature maps are Lipschitz continuous:

    \begin{theorem}[Lipschitz continuity]
        Let $x \in \mathbb{R}^d$. Let $\rho(x)=S(x)\rho_0S^\dagger(x)$, where $S(x) = \prod_{l=1}^L\prod_{k=1}^d S_{lk}(x_k)$, with $S_{lk}(\alpha)=\exp\left(-\frac{i}{2} \alpha H_{lk}\right)$ and $\lambda = \max_{k,l}\lVert H_{lk}\rVert_{\calH}$. Then, the map $x\mapsto\rho(x)$ is Lipschitz continuous with respect to the operator norm with a Lipschitz constant upper bounded by $\sqrt{dL}\lambda$:
        \begin{align}
            \lVert\rho(x)-\rho(x')\rVert_{\calH} &\leq \sqrt{dL}\lambda \lVert x-x'\rVert.
        \end{align}
    \label{thm:lipcontofrhomultidim}
    \end{theorem}
    \begin{proof}[Proof sketch]
        The Lipschitz constant is derived by bounding the derivatives of $\rho(x)$, which in turn depend on the derivatives of $S(x)$.
        The nested product structure of $S(x)$ facilitates explicit computation of these derivatives.
        Full proof details are provided in Appendix~\ref{a:proofthmlipcont}.
    \end{proof}

    Second, while the assumption of i.i.d. Gaussian inputs is under active study, there is a broad agreement on its plausibility for linear models in large-dimensional settings (large $p$)~\cite{Montanari2017universality, couillet2022random, pesce2023gaussiandata, adomaityte2024classification}.
    For instance, Ref.~\cite{seddik2020random} provides an intuitive justification using  \emph{Generative Adversarial Network} (GAN) models \cite{goodfellow2020GAN}.
    GANs generate remarkably realistic images by using Gaussian random noise as input.
    Under the assumption that GANs accurately replicate real-world data, it follows that real-life data can be approximated by Gaussian noise mapped through a Lipschitz continuous transformation. 
    Since Theorem~\ref{thm:lipcontofrhomultidim} establishes Lipschitz continuity for a broad class of quantum feature maps, the corresponding results from classical theory naturally extend to quantum models employing these maps.
    Accordingly, a (Q)ML model applied to such data can be conceptualized as a three-step process: (1) random Gaussian inputs, (2) hypothetical GAN-like feature map, and (3) the actual (Q)ML model.
    This perspective, further formalized in Ref.~\cite{couillet2022random}, supports the application of Theorem~\ref{thm:testpeak} to the empirical setting explored next in Section~\ref{s:numericalexperiments}.
    Hence, Theorem~\ref{thm:testpeak} can be broadly applied to predict a peak in test error in quantum linear models when the training set size $N$ matches the feature map dimension $p=4^n$.
    
\section{Empirical evidence of double descent in quantum kernel methods}
\label{s:numericalexperiments}

    \begin{figure*}[t!]
        \includegraphics[width=1.0\textwidth]{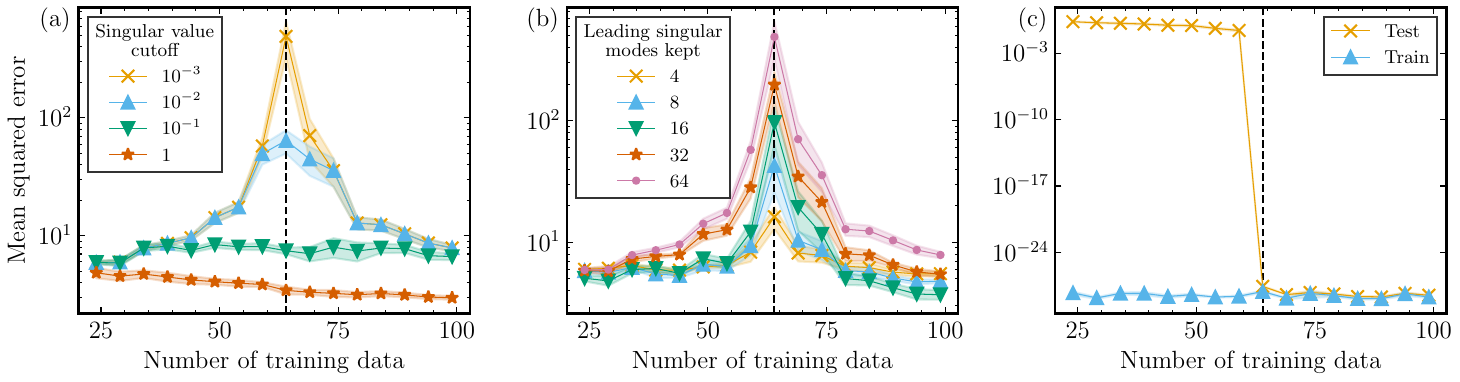}
        \caption{
            \textbf{Ablation experiments.} Mean test squared error as a function of the number of training data points $N$ for (a) applying different cutoffs on the minimum singular value, (b) retaining varying numbers of leading singular modes of the input data, and (c) eliminating residual error.
            All experiments were performed on the Synthetic dataset with $n=3$ qubits.
            The shaded area corresponds to the standard deviation for five independent experiment repetitions, each using independently sampled training data.
            The dotted black line indicates the interpolation threshold. The region to the left of the threshold corresponds to the overparameterized regime, while the region to the right corresponds to the underparameterized regime.
            Modifying each factor reduces or eliminates the double descent peak, highlighting their roles in the test error behavior.
        }
        \centering
        \label{fig:ablation_results}
    \end{figure*}
    
    The theoretical results presented in the previous section hold only in the asymptotic limit and rely on specific assumptions about the input data distribution.
    In this section, we validate our theory by providing empirical evidence of the double descent phenomenon in quantum kernel methods in the finite regime and with real-world data.
    To the best of our knowledge, this is the first empirical demonstration of double descent in QML.

    In our experiments, we consider QML models based on EQKs, as introduced in Eq.~\eqref{eq:eqk}.
    We note that the formalism laid out in Ref.~\cite{schaeffer2023doubledescentdemystifiedidentifying}, which we extend to quantum feature maps in Section~\ref{s:preliminaries}, closely resembles the well-established framework of QKMs~\cite{schuld2021supervised}.
    Specifically, the solution to the linear regression problem based on least squares (Eq.~\eqref{eq:LSsolution}) and minimum norm (Eq.~\eqref{eq:MNLSsolution}) can be recovered with kernel ridgeless regression~\cite{scholkopf2002learning}.
    Here, the term \enquote{ridgeless} denotes the absence of a regularization term in the conventional kernel regression formulation.
    Therefore, implementing ridgeless kernel regression with the EQK is mathematically equivalent to finding the optimal observables $\mathcal{M}^o$ and $\mathcal{M}^u$ as presented in Section~\ref{s:preliminaries}.
    While this connection to regression is clear, we also extend our experiments to classification tasks to further validate the theory.
    
    We perform a series of numerical experiments on three different datasets.
    The first dataset is synthetic, with dimensionality matching the number of qubits $n$.
    The target function is defined as $g(x) = \langle w, x \rangle$ with fixed custom weights $w$ and data points $x$ sampled uniformly from $[-\frac{\pi}{2}, \frac{\pi}{2}]^n$.
    Additionally, we utilized the Fashion MNIST dataset in a binary classification setting by choosing two categories.
    Lastly, for regression analysis on real-world data, we employed the California Housing dataset.
    Both of these datasets are accessible through Ref.~\cite{pedregosa2011scikit}.
    The quantum circuit simulations were performed using the {\tt PennyLane}~\cite{bergholm2018pennylane} software library running on classical hardware.
    The quantum kernel was implemented via an embedding feature map as defined in Eq.~\eqref{eq:eqk}, where each input $x$ is encoded into a quantum circuit comprising single-qubit rotations and entangling gates, ensuring the resulting Gram matrix at interpolation is full rank.
    For further implementation details, refer to Appendix~\ref{a:numerics} or the public code repository~\cite{github_repository}.
    
    To identify the interpolation threshold across different scenarios, we fixed the system size $n$ and varied the number of training samples $N$.
    Note that in our setting, the degree of parameterization can be determined by the ratio of the number of training data to the dimension of the feature space.
    Consequently, varying the number of training samples is somewhat equivalent to varying the feature space dimension.
    The test error curves, depicted in Fig.~\ref{fig:dd_qkernel}, exhibit a clear double descent peak.
    Notably, a sharp peak is observed at interpolation ($N/p=1$), followed by a characteristic drop in test error in the overparameterized regime ($N / p < 1$).
    This pattern is observed consistently across all datasets and system sizes, even though the data is far from Gaussian, as required by Theorem~\ref{thm:testpeak}. This empirical observation suggests that the conclusions drawn under the Gaussianity assumption remain robust in practice, as discussed in Section~\ref{s:spectralanalysis}.
    These findings provide compelling empirical evidence for the presence of double descent in QKMs and highlight its importance as a defining feature of their performance.
    
\section{Ablation study}
\label{s:ablation}
    
    To further solidify that the three factors outlined in Eq.~\eqref{eq:dd_error} directly influence the double descent in QKMs, we conduct a series of numerical experiments where we artificially ablate one factor at a time, inspired by the methodology in Ref.~\cite{schaeffer2023doubledescentdemystifiedidentifying}.
    The results are presented in Fig.~\ref{fig:ablation_results}.

    \paragraph{Singular values:}
    The factor $1/\sigma_r$ captures how the test error scales inversely with the singular values.
    In each experimental run, after sampling a training dataset of vectorized encoded states, we pre-processed the data by applying an explicit cutoff.
    Specifically, singular values below a specified threshold were removed from the singular matrix of the sampled data.
    As shown in Fig.~\ref{fig:ablation_results}(a), smaller cutoff values result in a more pronounced peak, emphasizing the role of small singular values in amplifying the test error.
    
    \paragraph{Test features:} The factor $\Tr{\rho_{r}^{V}\rho^{\vphantom{V}}_t}$ quantifies the overlap between the quantum feature space, spanned by the right singular vectors, and the test data features.
    To isolate this contribution, we pre-processed the test dataset by projecting it onto a subset of (leading) singular modes derived from the underlying distribution.
    This procedure provides a controlled mechanism to reduce the overlap between $\rho_{t}$ and the set $\{\rho_{r}^{V}\}$.
    Only the terms where projected test features have overlap with the training features contribute non-zero terms in the sum in Eq.~\eqref{eq:dd_error}.
    As illustrated in Fig.~\ref{fig:ablation_results}(b), retaining more modes in the pre-processing subset emphasizes the double descent peak.

    \paragraph{Residual error:} If the best model within the hypothesis class incurs no residual error on the training data, the term $\langle u_r, E\rangle$ vanishes, effectively eliminating the double descent peak.
    To test this, we explicitly used the underlying distribution in the fitting process.
    The results, shown in Fig.~\ref{fig:ablation_results}(c), demonstrate that under these conditions, the model perfectly fits the test data at interpolation and in the underparameterized regime, causing the test error to drop to zero. This is in agreement with the error decomposition in Eq.~\eqref{eq:overp_SVD}, where the second term only appears in the overparameterized regime. Note that this does not imply that the underparameterized regime is preferable in general. It only is preferable in settings where the ground truth model is known to lie within the hypothesis class (and the data is noise-free), which is rarely the case in realistic scenarios.

    These ablation experiments collectively demonstrate how all the three factors interact to produce the observed double descent behavior.
    Intuitively, when test features have a significant presence in the feature dimensions that are poorly represented by the training dataset (corresponding to small singular values of the training data), the model must extrapolate to fit these underrepresented dimensions~\cite{Jason2022Memorizing}. This extrapolation results in higher test error.
    Taken together, these results provide a clearer understanding of double descent in quantum models and emphasize the importance of carefully managing these factors during model design and training.

\section{The case of projected quantum kernels}\label{s:projected}

       \begin{figure*}[t]
        \centering
        \includegraphics[width=\textwidth]{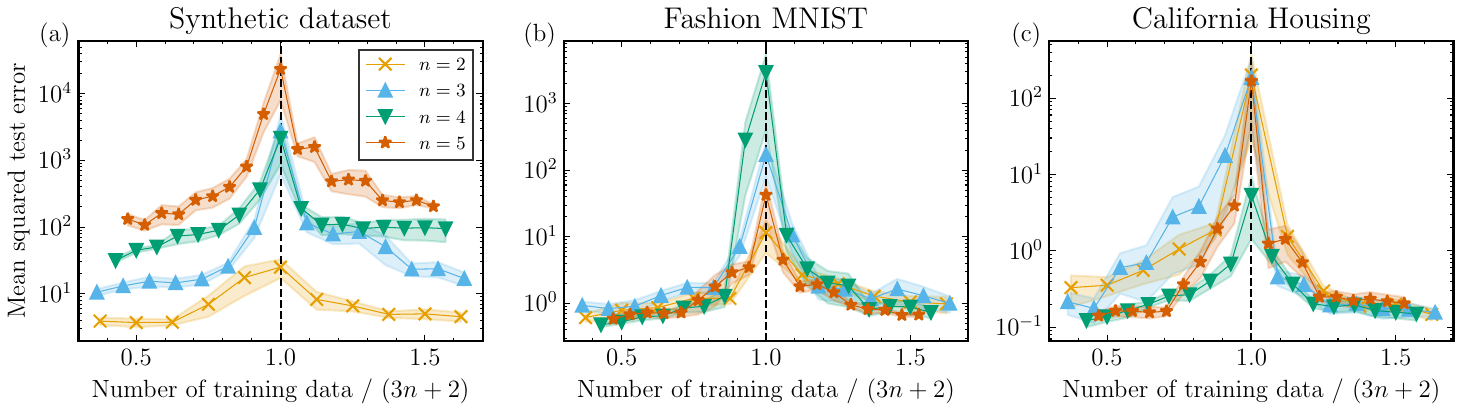} 
        \caption{
            \textbf{Empirical evidence of double descent in RDM projected quantum kernels.} Mean squared test error as a function of the normalized number of training data points $\left(N/(3n+2)\right)$ for the (a) Synthetic, (b) Fashion MNIST, and (c) California Housing datasets, employing the RDM kernel in Eq.~\eqref{eq:1rdm}.
            The shaded area corresponds to the standard deviation for five independent experiment repetitions, each using independently sampled training data.
            The dotted black line indicates the interpolation threshold. The region to the left of the threshold corresponds to the overparameterized regime, while the region to the right corresponds to the underparameterized regime.
            The consistent peak across all datasets and system sizes $n$ confirms the presence of double descent.
        }
        \label{fig:dd_pqk}
    \end{figure*}

    A natural question arises: can we control the location of the interpolation threshold in our experiments with QKMs? So far, we have successfully observed a peak in prediction error at the interpolation threshold, $p=N$, both analytically and empirically.
    However, we note that the quantum feature map we used in Section~\ref{s:numericalexperiments} resulted in an exponential-dimensional feature space with respect to the number of qubits $n$, specifically $p=4^n$.
    As a result, while our experiments provide evidence of double descent in quantum kernel methods, the peak in error occurs only for exponentially large training sets.
    We now introduce an alternative family of feature maps, referred to as \emph{projected feature maps}, whose feature space dimension (and hence the number of free parameters) is linear in the number of qubits $p\in\calO(n)$. Note that the exact value of $p$ for projected quantum kernels depends on the specific kernel design (here it is $3n+2$).     
    By employing these maps, we can effectively shift the interpolation threshold and facilitate a more controlled and efficient double descent setting.

    To construct a projected quantum feature state, we begin with a conventional quantum feature map $\rho$.
    This state is then transformed into a $n$-dimensional vector of single-qubit reduced density matrices (RDMs). This transformation is achieved using partial trace operations, also known as \emph{projections}:
    \begin{align}
    \label{eqn:rho_rdm}
        \rho_{\text{RDM}}(x) &= \left( \rho_{(k)}(x) \right)_{k=1}^{n} = \left( \Tr_{j \neq k} \{\rho(x)\}\right)_{k=1}^{n}.
    \end{align}
    Here, $\rho_{(k)}(x)$ represents the RDM corresponding to the $k^\text{th}$ qubit.
    Importantly, $\rho_{\text{RDM}}(x)$ forms a \emph{vector of matrices}, and we define their inner product as:
    \begin{align}
    \label{eqn:kappa_rdm}\langle\rho_{\text{RDM}}(x),\rho_{\text{RDM}}(x')\rangle &= \frac{1}{n}\sum_{k=1}^n\Tr\{\rho_{(k)}(x)\rho_{(k)}(x')\}\,.
    \end{align}
    Indeed, such projected feature maps also give rise to EQK functions on the Hilbert space of reduced density matrices $\kappa_{\text{RDM}}(x,x') = \langle\rho_{\text{RDM}}(x),\rho_{\text{RDM}}(x')\rangle$.
    We refer to this kernel as RDM kernel.
    Similar constructions have been explored in Refs.~\cite{huang2021power, gan2023unified}.
 
    In this section, we apply results from Sections~\ref{s:spectralanalysis} and~\ref{s:numericalexperiments} to the RDM kernel.
    Our primary focus lies on understanding the training set size at which the peak in test error occurs in this setting, both theoretically and empirically.
    We begin by extending Theorem~\ref{thm:testpeak} to apply to projected quantum feature maps.
    \begin{corollary}[Test error peak for projected feature maps]\label{cor:testpeakproj}
        Consider a quantum kernel model based on the RDM projection of an $L$-Lipschitz continuous quantum feature map $\rho(x)$ on $n$ qubits.
        Let $\{x_i\}_{i=1}^N$ be a training set of $d$-dimensional i.i.d. Gaussian random samples $x_i \sim \mathcal{N}(0, \mathbb{I}_d)$.
        Then, the training set size $N$ at which the peak in test error occurs with high probability is linear in the number of qubits; $N\in\calO(n)$.  
    \end{corollary}
    \begin{proof}
        To establish this result, we prove the following: (1) the projected feature map is Lipschitz continuous, and (2) the number of linearly independent dimensions $p$ satisfies $p\in\calO(n)$.
        For (1), the Lipschitz continuity of the projected feature map follows directly from the fact that the partial trace operation is Lipschitz continuous and that the composition of Lipschitz-continuous maps is also Lipschitz continuous.
        For (2), consider the linear dimension of the set of $1$-qubit states.
        The RDM kernel function defined in Eq.~\eqref{eqn:kappa_rdm} is an inner product in a $4n$-dimensional complex vector space:
        \begin{align}
            \kappa_\text{RDM}(x,x') &= \frac{1}{n}\sum_{k=1}^n\Tr\{\rho_{(k)}(x)\rho_{(k)}(x')\}\\
            &= \sum_{k=1}^n\sum_{i,j\in\{0,1\}} \frac{\left[\rho_{(k)}(x)\right]_{i,j}}{\sqrt{n}}\frac{\left[\rho_{(k)}(x')\right]_{i,j}^\ast}{\sqrt{n}}.\label{eq:1rdm}
        \end{align}
        Although the summation contains $4n$ terms, not all terms are independent.
        The unit-trace property of each reduced density matrix removes approximately $n$ terms. Consequently, the number of independent real components scales as $p\in\calO(n)$, and so the interpolation threshold appears at $N\in\calO(n)$.
    \end{proof}

    We repeat our numerical experiments using the RDM kernel, and the results are presented in Fig.~\ref{fig:dd_pqk}.
    Once again, we observe the characteristic double descent peak, now arising from projected quantum feature states.
    These results can be directly compared to Fig.~\ref{fig:dd_qkernel}, which corresponded to the general embedding quantum feature map.
    Notably, the training set size $N$ at which the test error peak occurs now scales linearly with the number of qubits $n$ when using the RDM kernel.
    This linear behavior contrasts sharply with the exponential scaling observed in Section~\ref{s:numericalexperiments}.
    These empirical findings strongly align with the theoretical predictions.

\section{Discussion} \label{s:discussion}

    The double descent phenomenon is a leading concept in understanding the success of deep learning, demonstrating how large models can outperform predictions from traditional statistical learning theory.
    In this work, we demonstrate that quantum kernel methods can also exhibit double descent.
    We present analytical results derived from a combination of linear regression insights and random matrix theory to understand the occurrence of the characteristic double descent behavior.
    Furthermore, we test our theory in extensive numerical experiments on synthetic and real-world datasets.
    To the best of our knowledge, this is the first study to identify this phenomenon in quantum machine learning (QML), marking a significant step forward in understanding the behavior of overparameterized quantum models.
    We now discuss the implications of these results and propose directions for future research.

    Our findings contribute to the ongoing developments in QML by aligning with and expanding the scope of existing studies.
    To begin with, our work is in line with the insights from Ref.~\cite{gil2024understanding}, which advocates for a paradigm shift in assessing the performance of QML models, moving beyond conventional statistical learning theory.
    Moreover, Ref.~\cite{peters2023generalization} presents a tailored setting of large PQC-based QML models that show improved test errors compared to statistical learning theory predictions.
    By introducing the double descent formalism, our work provides a complete complementary perspective, supported by results for common models on real-world data.
    Furthermore, although not directly contradicting Ref.~\cite{du2023problem}, our demonstration of double descent in quantum kernel methods provides a complementary perspective to their findings, which highlight that \emph{quantum classifiers undergo a U-shaped risk curve, in contrast to the double-descent risk curve of deep neural classifiers}.
    %Furthermore, the confirmed existence of double descent in QML seems to challenge statements put forth in Ref.~\cite{du2023problem}, where it is suggested that \emph{quantum classifiers undergo a U-shaped risk curve, in contrast to the double-descent risk curve of deep neural classifiers}.
    To our understanding, the conclusions in Ref.~\cite{du2023problem} may stem from specific scenarios where the model fails to fit the training set adequately at scale, resulting in large test errors.
    
    The framework presented applies to a broad class of linear models in quantum feature space, including certain quantum kernel methods~\cite{schuld2021supervised}, as well as specific extreme learning machines~\cite{innocenti2023potential, xiong2024fundamentalaspectsquantumextreme}.
    An important open question is whether similar double descent dynamics emerge in parameterized quantum circuits (PQCs), including those involving data re-uploading~\cite{perez2020data}.
    Recent work~\cite{thabet2024quantum} suggests that PQCs do not always converge to the minimum-norm solution in the overparameterized regime.
    Thus, further investigation is required to adapt the double descent framework to PQCs, a direction we leave for future work.

    %A widely recognized challenge in scaling quantum kernel methods is the exponential concentration problem, where the similarity measure between data points diminishes exponentially.
    %- we are studying things in the fixed
    %Our findings on double descent shall be considered in light of the exponential concentration problem in quantum kernels~\cite{thanasilp2024exponentialconcentrationquantumkernel}.
    Another important aspect to discuss is the connection of our work to the concentration phenomenon observed in quantum kernel methods~\cite{thanasilp2024exponentialconcentrationquantumkernel}.
    Our work does not solve concentration. Indeed, exponentially concentrated kernels with flat spectra will yield poor generalization in practice regardless. It is however interesting to understand the test performance expected from quantum kernels that do not suffer from this phenomenon. For such viable, non-exponentially-concentrated kernels, our analysis predicts the potential for a double descent curve (Fig.~\ref{fig:4step}), suggesting that test error may not always increase monotonically with complexity. Although this work focuses on the peak, the practical benefit hinges on achieving a deeper second dip in test error. Notably, a promising observation is that, in classical machine learning, the second dip is usually deeper~\cite{belkin2019reconciling, nakkiran2021deep}.

    On a broader level, our work establishes a foundational understanding of double descent in QML models, laying the groundwork for further exploration into the overparameterized regime and its implications for learning performance in quantum machine learning. \\

\subsection*{Code and data availability}
    The code and data used in this study are available in GitHub~\cite{github_repository}.

\subsection*{Acknowledgments}
    The authors thank Zo\"e Holmes, Evan Peters, Franz Schreiber, and Supanut Thanasilp for useful comments on a previous version of this manuscript.
    A.I. thanks Juan Carrasquilla, Hsin-Yuan Huang, Evan Peters, and Jason Rocks for helpful discussions.  
    A.I. also acknowledges support from the Natural Sciences and Engineering Research Council (NSERC), the Shared Hierarchical Academic Research Computing Network (SHARCNET), Compute Canada, and the Canadian Institute for Advanced Research (CIFAR) AI chair program.
    Resources used in preparing this research were provided, in part, by the Province of Ontario, the Government of Canada through CIFAR, and companies sponsoring the Vector Institute.  
    E.G.-F. is a 2023 Google PhD Fellowship recipient and acknowledges support by the Einstein Foundation (Einstein Research Unit on Quantum Devices), BMBF (Hybrid), and BMWK (EniQmA), as well as travel support from the European Union’s Horizon 2020 research and innovation programme under grant agreement No 951847.    
    \mbox{C.B.-P.} acknowledges support by the BMBF (MUNIQC-Atoms).    
    E.v.N. acknowledges support from the Dutch National Growth Fund (NGF), as part of the Quantum Delta NL programme.    
    V.D. acknowledges support from the Dutch Research Council (NWO/OCW), as part of the Quantum Software Consortium programme (project number 024.003.03), and co-funded by the European Union (ERC CoG, BeMAIQuantum, 101124342).

\subsection*{Disclaimer}
    The results, opinions, and conclusions expressed in this publication are not necessarily those of Volkswagen Aktiengesellschaft or those of the European Union or the European Research Council.
    None of the parties involved can be held responsible for them.

\bibliography{references}

\vspace*{0.5cm}

\onecolumngrid
\appendix

\newpage

\setcounter{figure}{0}
\setcounter{theorem}{0}

\begin{center}
\large{Supplementary Material for ``Double descent in quantum kernel methods''
}
\end{center}
\counterwithin{figure}{section}

\section{Prior and related work}
\label{a:related_work}

    In this section, we recap existing works that study the generalization performance of large quantum models and analyze how our work connects with them.
    We also clarify various definitions of overparameterization that are used in the quantum machine learning (QML) literature.
    To avoid confusion, we first list the various types of QML models generally considered in supervised QML literature.
    (1) Implicit quantum models are linear models that output $ \bra{0} U^{\dagger}(x) \calM U(x) \ket{0} $, where $U(x)$ is the fixed encoding unitary.
    Observable $\calM$ is a linear combination of the encoded training data, whose coefficients can be solved for classically.
    (2) Explicit quantum models are linear models that output $ \bra{0} U^{\dagger}(x) \calM U(x) \ket{0} $, where $U(x)$ is the fixed encoding unitary.
    Observable $\calM$ is given by $V^{\dagger}(\theta)O V(\theta)$, where $V(\theta)$ represents parameterized unitaries and $O$ is a fixed observable.
    The circuit parameters are generally optimized classically via gradient-based methods.
    (3) Data re-uploading models are linear models that output $\bra{0} W^{\dagger}(x, \theta) O W(x,\theta) \ket{0}$, where $W(x,\theta)= \prod_{l=1}^{L} V_l(\theta) U_l(x)$ is composed of alternating (and potentially distinct) encoding and variational layers.
    The encoding part consists of fixed gates and the trainable parameters are optimized classically.
    In general, the encoding unitary in any of these models can be made variational as well.
    This is especially useful in metric and representation learning applications \cite{lloyd2020quantumembeddings}.
    
    Various machine learning phenomena that connect model performance to its complexity include overfitting, benign overfitting, and double descent.
    Overfitting~\cite{ShalevShwartz2014understanding} occurs when a learning model performs well on the training data but fails to perform well on unseen test data.
    It can generally be attributed to high model complexity, although excessive data noise or excessive training can also lead to it.
    Benign overfitting refers to the phenomenon where interpolating models (with nearly-zero training error) can still generalize effectively~\cite{Muthukumar2019Harmless, Bartlett2020benign, Hastie2022surprises}.
    It is generally not concerned with when and how the interpolation regime begins and what are the necessary conditions to achieve it.
    In contrast, double descent~\cite{belkin2019reconciling} captures the non-monotonic behavior of the test error as model complexity varies, with the test error decreasing after an initial peak.
    Importantly, double descent is not inherently tied to or caused by interpolation; it can occur even in models that do not achieve zero training error.
    Furthermore, the point at which test error begins to descend for the second time is not strictly determined by the location of the interpolation threshold~\cite{curth2023Uturn}.
    
    Overfitting in QML models has been studied in several works using the conventional bias-variance tradeoff argument, although these works have mostly been limited to data re-uploading and explicit quantum models \cite{caro2021encoding, banchi2021generalization, chen2021ontheexpressibility, du2023problem, gyurik2023structural, peters2023generalization, haug2024generalization}.
    Various strategies have been proposed to overcome overfitting in these models, including qubit dropout \cite{schuld2020circuitcentric}, gate dropout \cite{kobayashi2022overfitting, scala2023ageneral} and parameter pruning \cite{Sim2021adaptivepruning, Haug2021capacity, Ohno2024adaptivepruning}.
    
    To our knowledge, our work is the first empirical demonstration of the double descent phenomenon in QML.
    The question of why double descent has not been previously observed for quantum models remains open.
    One possible explanation is that QML is mirroring the historical development of classical machine learning (ML).
    With the onset of the use of heavily parameterized models and high-dimensional data, large classical models did not always exhibit benign overfitting, nor was double descent and an interpolation peak consistently observed.
    Recent developments show that this could be due to how generalization performance was measured \cite{curth2024classical} and how bias and variance were conventionally defined \cite{Jason2022Memorizing}, hinting that the traditional bias–variance tradeoff argument based on model complexity breaks down in the highly overparameterized regime \cite{dar2021farewellbiasvariancetradeoff}.
    
    It is possible that QML literature is following a similar path.
    Some recent works have highlighted that most generalization bounds derived in the QML literature become vacuous in the overparameterized regime \cite{du2023problem, gil2024understanding}.
    For example, Ref.~\cite{chen2021ontheexpressibility} studied overfitting in parameterized quantum circuits, showing that the expressivity and VC dimension of these quantum circuits saturate with increasing circuit depth.
    However, it has been shown in both classical~\cite{curth2023Uturn} and quantum~\cite{gil2024understanding} settings that VC dimension is not a reliable measure to study generalization when scaling model complexity.
    Another example is Ref.~\cite{du2023problem} which proposed that large QML models cannot achieve low training errors while, in contrast, Ref.~\cite{gil2024understanding} showed that some quantum models can perfectly fit training data even with label noise and Ref.~\cite{peters2023generalization} presented scenarios where large quantum models can perform well and display benign overfitting.
    We look at Ref.~\cite{peters2023generalization} and Ref.~\cite{du2023problem} in more detail below.
    The takeaway message of these recent developments is that there is a need to redefine how we connect model complexity to generalization as we scale the data dimension and model complexity, which we leave to follow-up work.

    Benign overfitting (or harmless interpolation) phenomenon was investigated in explicit quantum models in Ref.~\cite{peters2023generalization} using spiking models, which is one of the established techniques in classical ML literature.
    To perfectly fit training points while allowing for general prediction, this technique allows the predictor function to have a \enquote{spiked} localized behavior near training points while being smooth in other areas of the input domain.
    By using the Fourier representation of the functions in the hypothesis class, trigonometric polynomial interpolation is employed to formulate the interpolation problem.
    Hence, overparameterization occurs when the number of Fourier modes of the learning model exceeds the number of Fourier modes of the data-generating distribution.
    The authors first show how the feature weights of an overparameterized Fourier features model can be engineered to have benign overfitting for linear regression.
    Then, they extend the classical results by using the Fourier representation of quantum models \cite{schuld2021effect}.
    This allows them to design quantum models with similar features by tuning the initial quantum states and data-encoding Hamiltonians.
    Hence, Ref.~\cite{peters2023generalization} provides important tools for designing and working with overparameterized quantum models in the harmless interpolation regime.
    While the emphasis of their analysis is not on the precise location of the interpolation threshold, signatures of double descent are evident in their numerical experiments (Fig.~3).
    The authors use synthetic data and design a specific feature map that ensures the appropriate spiking behavior.
    In turn, in our work, we studied the factors that contribute to the peak at interpolation, relying on random matrix theory and the spectral properties of the data matrix to explain it.
    Consequently, our framework works for real-life data, with a non-restrictive generic quantum feature map.
    We leave a thorough study on the similarities between the underlying generalization mechanisms for future work.
    
    Ref.~\cite{du2023problem} studies overfitting and the potential presence of double descent in classification with quantum models.
    The authors use explicit quantum models as quantum classifiers and define overparameterization as the regime when the size of training dataset exceeds the number of data-independent parameterized quantum gates and the circuit unitary becomes a 2-design.
    By considering arbitrary circuit ansatzes (without problem-informed structure), where overparameterized circuits become 2-designs, training is likely to suffer from barren plateaus \cite{larocca2024review}.
    This will result in a U-shaped curve for the total expected error (training error + generalization error), as poor trainability will lead to larger training errors (and hence poor generalization error) as more parameters/gates are added.
    Hence, the authors propose using underparameterized problem-dependent ansatzes instead.
    They reinforce this via numerical studies conducted using hardware-efficient ansatz on the Fashion MNIST dataset, which fails to reach low training errors due to poor trainability.
    This particular setting and the corresponding poor performance of explicit quantum models can also be studied from the generalization perspective (independent of the trainability perspective).
    Along these lines,  Ref.~\cite{peters2023generalization} showed that using states sampled uniformly from the Haar measure will have concentrated feature weights for most known encoding strategies (see their Appendix B.1).
    This will give rise to quantum models that will not achieve benign overfitting when overparameterized.
    However, we note that these results should not rule out the possibility of double descent in all problem-agnostic quantum models (classifiers or otherwise), especially those that do not suffer from barren plateaus, like QCNNs.
    We leave a complete formal study of QCNNs for follow-up work.
    
    Ref.~\cite{du2023problem} also proposes that in order to reach zero training errors, parameterized quantum classifiers should optimize parameters such that the optimal feature states and measurements form a general simplex equiangular tight frame
    (ETF).
    ETFs have been shown to arise in deep classical neural networks as they reach zero training error (when trained for long times), leading to ``neural collapse'' (small variability in the state representation of training data that belong to the same class) \cite{papyan2020prevalence}.
    However, we note that while ETFs and neural collapse (on training data) can mean low training errors, it does not necessarily imply low generalization errors and may even lead to worse performance on test data, as was recently studied in \cite{hui2022limitations}.
    So, neural
    collapse is an optimization phenomenon, not a generalization one.
    This could explain the numerical results in Ref.~\cite{du2023problem} where they use hardware-efficient ansatz to do binary classification on the parity dataset (Fig.~2 and Fig.~J9 in Appendix~J) where only harmful overfitting is observed for the overparameterized classifier.
    Hence, while Ref.~\cite{du2023problem} studies an important aspect of how deep quantum models might behave and represent the data during training, there is a need to connect this to phenomena like benign overfitting and double descent, which we leave for follow-up work.
    
    Other definitions of overparameterization have also been reported in QML literature.
    As an important clarification, we reiterate that Ref.~\cite{Haug2021capacity} and Ref.~\cite{larocca2023theory} use tools from quantum control and quantum information theory to define overparameterization in parameterized quantum circuits.
    These large quantum models have been reported to train faster \cite{kiani2020learningunitaries, wiersema2020exploring, kim2021universal, Grimsley2023adaptive} and show noise resilience \cite{fontana2021evaluating, duschenes2024characterization} in certain settings.
    For instance, a ``computational phase transition'' at the critical point between the under-parameterized and over-parameterized regimes was shown in Ref.~\cite{kiani2020learningunitaries}, where the task is to learn a Haar random unitary.
    The faster convergence and changes in loss landscape to global minima are well understood in classical deep learning.
    We leave the development of a formal connection between the conventional ML definition of overparameterization to that defined using quantum control for follow-up work.
    
    Another important clarification to make is how wide quantum neural network (QNN) literature relates to our work.
    Wide QNNs represent the quantum analog of classical neural networks that are studied in the limit of infinite width (number of hidden layer units).
    This hypothetical setting makes it easier to study the training dynamics of the network, which can then be used to understand the behavior of large models.
    Like their classical counterparts, wide QNNs show a \enquote{lazy training} regime where for wide enough networks, the initial parameters do not change a lot during training \cite{liu2022representationlearning, shirai2024qtk} and the network displays exponentially fast convergence (training
    error decreases exponentially with iterations) to a good local minimum \cite{you2022convergencetheory, Liu2023analytictheory, you2023analyzingconvergence}.
    However, in our work, we use only non-deep overparameterized models -- which are also more practical.
    Hence, we are not only reporting the first empirical proof of double descent in QML but also the first such observation for a non-deep quantum learning model.
    This is not different from classical ML where double descent is ubiquitous in deep learning but has also been observed in many non-deep learning methods like boosting, trees, and linear regression \cite{belkin2019reconciling}.

\section{Remark on the parameter count}\label{a:parametercount}

    In this section, we give further details on the exact number of parameters $p$ in quantum kernel methods. When embedding quantum kernels via a feature map $\rho(x)$ on $n$ qubits, a key quantity is the number of linearly independent real parameters in the representation. This number is upper bounded by the dimension of the vector space of Hermitian matrices acting on a $2^n$-dimensional Hilbert space.

    A general complex matrix of size $d \times d$ has $d^2$ complex entries, corresponding to $2d^2$ real parameters. However, quantum states are represented by density matrices, which are Hermitian, positive semi-definite, and trace-normalized. Hermiticity alone imposes that the matrix entries satisfy $\rho = \rho^\dagger$, which reduces the number of real parameters from $2d^2$ to $d^2$.

    This is seen by decomposing any Hermitian matrix $M$ as $M = S + iA$, where $S$ is a real symmetric matrix and $A$ is a real antisymmetric matrix with zero diagonal. The number of independent real parameters in this decomposition is $\frac{d(d+1)}{2}$ from $S$, and $\frac{d(d-1)}{2}$ from $A$. Adding these yields $d^2$ real degrees of freedom. For $n$ qubits, where $d = 2^n$, the dimension of the Hermitian vector space is therefore $d^2 = 4^n$.

    While positivity and trace normalization constrain the valid subset of Hermitian matrices that represent physical density matrices, they do not reduce the dimensionality of the linear space in which these matrices reside. That is, although the set of density matrices is a strict subset of the Hermitian space, it still spans a space of dimension $4^n$ in terms of linearly independent real parameters. In this linear sense, the dimension remains unaffected by trace normalization, as can be seen in a simple one-qubit example:
    \begin{align*}
    \rho = \begin{pmatrix} p & a \\ a^* & 1-p \end{pmatrix},
    \end{align*}
    where $p$ and $1-p$ are linearly independent (the only solution to $a\cdot p + b\cdot (1-p) = 0$ that holds for all $p\in[0,1]$ is the trivial one $a=b=0$).

    The density-matrix representation of a pure state can be fully specified by $2^n$ complex coefficients, instead of $4^n$, owing to the fact that pure states correspond to rank-$1$ matrices.
    However, we would like to re-emphasize that the quantity of primary relevance to a linear regression problem is the number of \emph{linearly independent} real parameters.
    This means when considering the feature map $\rho(x)$, which is a mapping of input data to quantum states, the number of linearly independent components remains $4^n$, provided that the image of the feature map spans the full space of density matrices. That is true whenever the real components of the density matrix are linearly independent, i.e., $\sum_{i,j} a_{ij} \rho_{ij}(x) = 0$ has no solution \emph{independent of $x$} except for the trivial one $a_{ij}=0$ for all $i$ and $j$.
    Accordingly, the number of linearly independent real parameters in quantum kernel methods scales as $4^n$, reflecting the full dimensionality of the Hermitian space on $n$ qubits, which remains the relevant quantity for assessing expressivity in linear regression tasks.

\section{Detailed derivation of test error decomposition}
\label{a:regression_detail}

    In this section, we provide a more detailed derivation of the results obtained in Sec.~\ref{s:preliminaries}.
    
    Let the training dataset be represented as $ \calD = \{(\rho_i, y_i)\}_{i=1}^N$, where $\rho_i$ are density matrices.
    Using this, the data matrix $D$ and the label matrix $Y$ are defined as:
    \begin{align}
        D = \begin{pmatrix} \left[\rho_1\right]^\dagger \\ \vdots \\ \left[\rho_N\right]^\dagger \end{pmatrix}, \quad\quad Y = \begin{pmatrix} y_1 \\ \vdots \\ y_N\end{pmatrix}.
    \end{align}
    Here, $D$ is a collection of self-adjoint matrices, with each $\rho_i$ satisfying $\rho^\dagger=\rho$.
    We retain the Hermitian conjugate notation to maintain clarity, distinguishing between inner and outer products in the formalism.
    Specifically, in analogy to the notation for vectors and dual vectors (co-vectors), given a matrix $\rho$ and its \enquote{co-}matrix $[\sigma]^\dagger$, we adopt the following conventions:
    \begin{itemize}
        \item The inner product is the application of a co-matrix on a matrix: $[\sigma]^\dagger \rho = \langle \sigma,\rho\rangle_\text{HS}=\Tr\{\sigma\rho\}$, where we introduce the Hilbert-Schmidt inner product as the canonical inner product between Hermitian matrices.
        \item The outer product is the application of a matrix on a co-matrix: $\rho [\sigma]^\dagger = \rho\otimes\sigma$, where we introduce the tensor product as the canonical outer product between Hermitian matrices.
    \end{itemize}
    This way, the co-matrices are elements of the dual vector space of Hermitian matrices, so they are linear maps from Hermitian matrices to the reals.
    Together with the canonical basis of the dual space, both the inner and the outer product are well defined and follow the standard ones for usual vector spaces.
    
    The data matrix $D$ should therefore be interpreted as a vector of matrices, $D\in(\Herm(2^n))^N$, rather than a standard matrix.
    This distinction is essential for correctly defining inner and outer products within our framework.
    Meanwhile, the label vector $Y$ is a real-valued column satisfying $Y\in\bbR^{N\times 1}$.
    
    Since quantum models are linear in the Hilbert space of Hermitian operators, the goal is to identify optimal functions within the family
    \begin{align}
        \calG &= \{\Tr\{\rho\calM\}\,|\,\calM\in\Herm(2^n)\}, 
    \end{align}
    where $\calM$ is an observable or ``parameter matrix'' with $p = 4^n$ free parameters (dimension of the corresponding orthonormal Hermitian basis).
    We examine two regimes based on the relationship between the number of parameters $p$ and the number of training data points $N$: the \textit{underparameterized} regime, where $N > p$, and the \textit{overparameterized} regime, where $N < p$.
    In the underparameterized case, the number of data points exceeds the number of parameters while in the overparameterized case, the number of parameters exceeds the number of data points.
    In both scenarios, the task involves solving a quadratic optimization problem to determine the linear function that best fits the data.
    
    Let $D\calM$ denote the application of the data matrix $D$ to a parameter matrix $\calM$, defined analogously to scalar-vector multiplication but extended to matrices:
    \begin{align}
        D\calM &= \begin{pmatrix} \left[\rho_1\right]^\dagger \\ \vdots \\ \left[\rho_N\right]^\dagger \end{pmatrix}\begin{pmatrix} \calM \end{pmatrix}
        = \begin{pmatrix} \left[\rho_1\right]^\dagger\calM \\ \vdots \\ \left[\rho_N\right]^\dagger\calM \end{pmatrix} 
        = \begin{pmatrix} \Tr{\rho_1\calM} \\ \vdots \\ \Tr{\rho_N\calM} \end{pmatrix} .
    \end{align}
    In the underparameterized regime, the model parameters are determined using the standard least squares optimization problem:
    \begin{align}
        \calM^u &= \argmin_\calM \lVert D\calM - Y\rVert^2.
    \end{align}
    The solution to this problem is unique and takes the form~\cite{engl1996regularization}
    \begin{align}
            \calM^u &= \left(D^\dagger D\right)^{-1}D^\dagger Y,
    \label{eq:LSsolution_a}
    \end{align}
    where $D^\dagger D$ is the $p \times p$ sample covariance matrix (we abuse notation here: to be more precise, $D^\dagger D$ is the Hessian matrix, and $D^\dagger D/N$ is the sample covariance matrix).
    Expanding this term, we find:
    \begin{align}
        D^\dagger D &= \begin{pmatrix} \rho_1 & \rho_2 & \cdots & \rho_N \end{pmatrix}\begin{pmatrix} \left[\rho_1\right]^\dagger \\ \left[\rho_2\right]^\dagger \\ \vdots \\ \left[\rho_N\right]^\dagger \end{pmatrix} \\
        &= \rho_1\left[\rho_1\right]^\dagger + \rho_2\left[\rho_2\right]^\dagger + \cdots + \rho_N\left[\rho_N\right]^\dagger  \\
        &= \rho_1\otimes\rho_1 + \rho_2\otimes\rho_2 + \cdots + \rho_N\otimes\rho_N \\
        &= \sum_{k=1}^N \rho_k\otimes\rho_k  .
    \end{align}
    Let the corresponding optimal predictor function here be represented by $g^u \in \calG$, i.e., $g^u(\rho)= \Tr{\rho \calM^u }$.
    In contrast, for the overparameterized regime, infinitely many solutions satisfy the least squares condition.
    We therefore choose the standard minimum-norm least squares optimization problem~\cite{engl1996regularization}:
    \begin{align}
        \calM^o &= \argmin_\calM \lVert\calM\rVert^2_2 \\
        &\hphantom{=} \text{s.t.}\Tr{\rho_k\calM}=y_k, \forall k\in [N]
    \end{align}
    The optimal solution for this problem is well-known from optimization theory:
    \begin{align}
        \calM^o &= D^\dagger\left(DD^\dagger\right)^{-1} Y,
    \label{eq:MNLSsolution-a}
    \end{align}
    where $DD^\dagger$ is the $N \times N$ Gram matrix, explicitly written as:
    \begin{align}
            DD^\dagger &= \begin{pmatrix} \left[\rho_1\right]^\dagger \\ \left[\rho_2\right]^\dagger \\ \vdots \\ \left[\rho_N\right]^\dagger \end{pmatrix}\begin{pmatrix} \rho_1 & \rho_2 & \cdots & \rho_N \end{pmatrix} \\
            &= \begin{pmatrix} \left[\rho_1\right]^\dagger\rho_1 & \left[\rho_1\right]^\dagger\rho_2 & \cdots & \left[\rho_1\right]^\dagger\rho_N \\
            \left[\rho_2\right]^\dagger \rho_1 & \left[\rho_2\right]^\dagger \rho_2 & & \\
            \vdots & & \ddots & \\
            \left[\rho_N\right]^\dagger\rho_1 & & & \left[\rho_N\right]^\dagger\rho_N \end{pmatrix}  \\
            &=\begin{pmatrix} \Tr\{\rho_1\rho_1\} & \Tr\{\rho_1\rho_2\} & \cdots & \Tr\{\rho_1\rho_N\} \\
            \Tr\{\rho_2\rho_1\} & \Tr\{\rho_2\rho_2\} & & \\
            \vdots & & \ddots & \\
            \Tr\{\rho_N\rho_1\} & & & \Tr\{\rho_N\rho_N\} \end{pmatrix}  \\
            &= \left( \Tr{\rho_k\rho_l}\right)_{k,l=1}^N .
    \label{eq:grammatrix}
    \end{align}
    Here, we let the corresponding optimal predictor function be represented by $g^o \in \calG$, i.e., $g^o(\rho)= \Tr{\rho \calM^o }$.
    Having obtained the empirical risk minimizer for both regimes, it is important to note that these solutions are not guaranteed to coincide with the ground truth for all instances.
    To formalize this, let us introduce an oracle that provides the expected risk minimizer $\calM^\ast$, which represents the optimal linear model parameters for the entire data distribution:
    \begin{align}
        \calM^\ast &= \argmin_{\calM}\left\{\bbE_{(\rho,y)} \big[ \left(\Tr{\rho\calM} - y\right)^2 \big] \right\},
    \end{align}
    where the expectation value is taken over the underlying data distribution.
    Notably, even with the expected risk minimizer, the ground truth need not be a linear function in the feature space, implying that $\calM^\ast$ may still incur some error for individual points.
    To explore this further, consider a new test data point $(\rho_t, y_t)$ outside the training set.
    Note that this setting where the new test point lies outside the training dataset is called a \emph{random design} setting; both train and test data are separately sampled from the underlying distribution and the learned predictor does not encounter the test data during training.
    Note that this is an essential component to achieve double descent \cite{curth2024classical}.
    It will not occur in a \emph{fixed design} setting, where the same training inputs are used during testing but with resampled labels.
    
    In general, neither the empirical risk minimizer $\calM^{u,o}$ nor the expected risk minimizer $\calM^\ast$ will perfectly predict the label $y_t$:
    \begin{align}
        \Tr{\rho_t\calM^{u,o}} - y_t &\neq 0 \\
        \Tr{\rho_t\calM^\ast} - y_t &\eqcolon e_t.
    \end{align}
    Here, $e_t$ represents the ``residual'' error made by $\calM^\ast$ on the new input $\rho_t$ due to label noise or mismatch between the underlying data-generating function and the family of functions $\calG$.
    Specifically, for the training set $\calD$, we define $E$ as the vector of  errors:
    \begin{align}
    \label{eq:residual_bias_def}
        E &\coloneqq Y - D\calM^\ast \\
        e_k &\coloneqq y_k - \Tr\{\rho_k\calM^\ast\}.
    \end{align}
    Note that the expected value of this residual error (over the underlying distribution) is usually called the irreducible error in literature \cite{dar2021farewellbiasvariancetradeoff}.
    Following Ref.~\cite{schaeffer2023doubledescentdemystifiedidentifying}, we now investigate the difference in predictive performance between the empirical risk minimizer and the expected risk minimizer on unseen data.
    For clarity, we introduce the following notation for predictions:
        \begin{align}
            y^{u,o}_t &\coloneqq g^{u,o}(\rho_t)=\Tr{\rho_t\calM^{u,o}} \\
            y^\ast_t &\coloneqq g^\ast(\rho_t)=\Tr{\rho_t\calM^\ast}.
        \end{align}
    Conventionally, the test error in the predictor is defined as the mean squared error in the assigned labels compared to the underlying distribution.
    Here, we define the generalization error made by the empirical risk minimizer as the mean squared label error compared to the expected risk minimizer \cite{dar2021farewellbiasvariancetradeoff}:
        \begin{align}
        \label{eq:MSE_y_t_def}
        \text{MSE}[g^{u,o}]&= \bbE_{(\rho,y)} \mathbb{E}_{\mathcal{D}} \big[ \big( y^{u,o} - y^\ast \big)^2 \big] \\
        &= \text{Var}[g^{u,o}] + \text{Bias}^2[g^{u,o}],
        \end{align}
    where $\mathbb{E}_{\mathcal{D}}$ represents averaging over all possible training datasets.
    Here, bias captures how well the learned predictor from the class of functions performs compared to the optimal linear model.
    It is defined as:
        \begin{align}
        \text{Bias}[g^{u,o}]&= \bbE_{(\rho,y)} \bigg[ \bigg( \mathbb{E}_{\mathcal{D}} \big[ y^{u,o} \big] - y^\ast  \bigg)^2 \bigg].
        \end{align}
    Similarly, variance of the learned predictor function is defined as:
        \begin{align}
        \text{Var}[g^{u,o}]&= \bbE_{(\rho,y)} \mathbb{E}_{\mathcal{D}} \bigg[ \bigg( y^{u,o} - \mathbb{E}_{\mathcal{D}} \big[ y^{u,o} \big] \bigg)^2 \bigg], 
        \end{align}
    which captures the fluctuations in prediction across the choice of training dataset.
    This method of comparing the performance of empirical risk minimizer with respect to the expected risk minimizer helps to identify the various sources contributing to the test error.
    For a single test data point $(\rho_t, y_t)$ and a given sampled training dataset $\calD$, we can now analyze the difference $y^{u,o}_t-y^\ast_t$ in the test error by substituting the expressions for $\calM^{u,o}$ derived earlier in Eqs.~\eqref{eq:LSsolution_a} and~\eqref{eq:MNLSsolution-a}.
    For the underparameterized case, we have:
        \begin{align}
        \label{eq:a_yut-ystar}
            y^u_t - y^\ast_t &= \Tr{\rho_t(D^\dagger D)^{-1}D^\dagger Y} - \Tr{\rho_t\calM^\ast} \\
            &=\Tr{\rho_t(D^\dagger D)^{-1}D^\dagger (D\calM^\ast + E)} - \Tr{\rho_t\calM^\ast} \\     
            &=\Tr{\rho_t(D^\dagger D)^{-1}D^\dagger D\calM^\ast} + \Tr{\rho_t (D^\dagger D)^{-1}D^\dagger E} - \Tr{\rho_t\calM^\ast} \\
            &= \Tr{\rho_t\calM^\ast} + \Tr{\rho_t (D^\dagger D)^{-1}D^\dagger E} - \Tr{\rho_t\calM^\ast} \\
            &= \Tr{\rho_t (D^\dagger D)^{-1}D^\dagger E} .
            \end{align}
    Whereas for the overparameterized case, we have:          
        \begin{align}
        \label{eq:a_yot-ystar}
            y^o_t - y^\ast_t &= \Tr{\rho_t D^\dagger(DD^\dagger)^{-1}Y} - \Tr{\rho_t\calM^\ast} \\
            &= \Tr{\rho_t D^\dagger(DD^\dagger)^{-1}(D\calM^\ast + E)} - \Tr{\rho_t\calM^\ast} \\
            &= \Tr{\rho_t D^\dagger(DD^\dagger)^{-1}D\calM^\ast} + \Tr{\rho_t D^\dagger(DD^\dagger)^{-1}E} - \Tr{\rho_t\calM^\ast} \\
            &= \Tr{\rho_t\left(D^\dagger(DD^\dagger)^{-1}D - \bbI_{2^n\times2^n}\right)\calM^\ast} + \Tr{\rho_t D^\dagger(DD^\dagger)^{-1}E}.
        \end{align}
    Although the derived expressions may initially appear complex, they exhibit a clear and systematic structure that highlights the distinct mathematical roles played by the sample covariance matrix $D^\dagger D$, and the Gram matrix $DD^\dagger$, in the respective regimes.
    The data matrix $D \in \mathbb{C}^{N\times(2^n\times2^n)}$ is rectangular, meaning it does not have a direct inverse.
    To address this, square matrices are constructed by multiplying $D$ with its Hermitian conjugate either from the left, yielding the sample covariance matrix $D^\dagger D \in \mathbb{C}^{(2^n\times2^n)\times(2^n\times2^n)}$, or from the right, producing the Gram matrix $DD^\dagger \in \mathbb{R}^{N\times N}$.
    In the underparameterized regime, where $N > 4^n$, the rank of $D$ is at most $4^n$, making the sample covariance matrix $D^\dagger D$ the appropriate choice for constructing the inverse required for solving the least squares problem.
    This matrix contains the relationships between the parameters of the model and ensures the solution is well-defined.
    In contrast, in the overparameterized regime, where $N<4^n$, the data points are fewer than the model parameters, and the Gram matrix $DD^\dagger$ becomes the relevant object.
    This matrix captures the relationships among the training data points themselves, reflecting the nature of the optimization problem in this regime.
    This distinction between $D^\dagger D$ and $DD^\dagger$ underscores the fundamental interplay between model complexity and dataset size.
    Hence, if $D^+$ represents the pseudoinverse of the data matrix $D$, then the error expressions for the underparameterized case can be simplified as:
        \begin{align}
        \label{eq:yut-ystar-D+}
            y^u_t - y^\ast_t  &= \Tr{\rho_t (D^\dagger D)^{-1}D^\dagger E} \\
            &= \Tr{\rho_t D^+ E},
        \end{align}
        where $D^+ = (D^\dagger D)^{-1}D^\dagger$ is the left inverse of $D$ such that $ D^+D = \bbI_{2^n\times2^n}$.
        Similarly, for the overparameterized case, we get:
        \begin{align}
        \label{eq:yot-ystar-D+}
            y^o_t - y^\ast_t &= \Tr{\rho_t D^\dagger(DD^\dagger)^{-1}E} \\
            &\hphantom{=}+ \Tr{\rho_t\left(D^\dagger(DD^\dagger)^{-1}D - \bbI_{2^n\times2^n}\right)\calM^\ast} \\
            &= \Tr{\rho_t D^+E} \\
            &\hphantom{=} +\Tr{\rho_t\left(D^+ D - \bbI_{2^n\times2^n}\right)\calM^\ast}, 
        \end{align}
        where $D^+ = D^\dagger(DD^\dagger)^{-1}$ is the right inverse of $D$ such that $ DD^+ = \bbI_{N}$.
        These expressions, while insightful, remain somewhat cumbersome in their current form.
        To simplify them, we employ the Singular Value Decomposition (SVD) of $D$:
        \begin{align}
        \label{eq:svd_D_definition}
            D &= U \Sigma V^\dagger, 
        \end{align}
        where $U \in\bbU(N)$ is a unitary matrix of left singular vectors, associated with the \enquote{vector index} of $D$, $\Sigma \in \bbR^{N\times(2^n\times2^n)}$ is a diagonal matrix of singular values, with $\Sigma^k_{ij} \propto \delta_{k,(2^N(i-1)+j)}$, and $V \in\bbU(2^n\times2^n)$ is a unitary matrix of right singular vectors, associated to the \enquote{matrix indices} of $D$.
    
        Given the SVD decomposition of the data matrix $D$ in Eq.~\eqref{eq:svd_D_definition}, its pseudoinverse $D^+$ can be defined as:
        \begin{align}
            D^+ = V \Sigma^+ U^\dagger , 
        \end{align}
        where $\Sigma^+$ is the pseudoinverse of $\Sigma$ containing the reciprocals of the non-zero singular values on its diagonal.
    
        Substituting the simplified forms into Eqs.~(\ref{eq:yut-ystar-D+}) and (\ref{eq:yot-ystar-D+}), we obtain:
        \begin{align}
            y^u_t-y^\ast_t &= \Tr{\rho_t V \Sigma^+ U^\dagger E} \\
            y^o_t - y^\ast_t &= \Tr{\rho_t V \Sigma^+ U^\dagger E} \\
            &\hphantom{=} + \Tr{\rho_t\left(D^+D - \bbI_{2^n\times2^n}\right)\calM^\ast}.
        \label{eq:a_overp_SVD}
        \end{align}
        These two expressions show the contribution of a variance-like term and a bias-like term to the test error defined in Eq~(\ref{eq:MSE_y_t_def}).
        The second term in prediction error in the overparameterized regime captures the bias of this predictor because the overparameterized system is under-determined; rank $R$ cannot exceed $N$ and the information in the remaining $p-N$ dimensions cannot be fully captured~\cite{schaeffer2023doubledescentdemystifiedidentifying}.
        Along the same lines, the zero bias in the underparameterized predictor here is valid in this context since we are comparing it to a linear \enquote{teacher} model of the expected risk minimizer (the best linear model in the hypothesis class)~\cite{Jason2022Memorizing}, which is not necessarily a good predictor.
        Note that this bias term is sometimes also referred to as \emph{in-class bias}, and it should not be confused with the bias commonly discussed in traditional statistical learning theory, which typically decreases with increasing model complexity.  The latter, often termed \emph{misspecification bias}, arises only when comparing the expected risk minimizer to the ground truth function~\cite{dar2021farewellbiasvariancetradeoff}.
        %Hence, defining the test error in this way disentangles the different sources of bias observed in a learning model.
        %This bias differs from the bias that results from using a linear model to learn a non-linear data-generating function, for example.
        Now, to further illustrate why the first term is a variance-like term, we expand it as:
        \begin{align}
            \Tr{\rho_t V \Sigma^+ U^\dagger E} &= \sum_{r=1}^R \frac{1}{\sigma_r} \Tr{ \rho^V_r \rho_t}\langle u_r , E\rangle.
        \label{eq:a_dd_error}
        \end{align}
        where we recall $R=\min\{N,2^n\times2^n\}$, and the right singular vectors $\rho^V_r$ are Hermitian, positive semi-definite, and unit-trace matrices.
        The sum captures three key factors:
        \begin{enumerate}
            \item The reciprocals of the singular values $1/\sigma_r$.
            \item The interaction of $\rho_t$ with the basis of right singular vectors $\Tr{ \rho^V_r \rho_t}$.
            \item The projection of $E$ onto the left singular vectors $\langle u_r, E\rangle$.
        \end{enumerate}
    
        This decomposition highlights the distinct roles of the left and right singular vectors.
        The left singular vectors encode linear combinations over the training data indices, while the right singular vectors span the quantum feature space as an orthonormal basis.
        Intuitively, the principal directions corresponding to small singular values of a given training dataset $\calD$ represent \enquote{under-sampled} directions of the underlying data distribution for which the training dataset does not capture enough information.
        Consequently, if a test data point has a large projection onto one of these principal components ($ [ \sigma_r \approx 0 ] \land [ \Tr{ \rho^V_r \rho_t} > 0] $), the model is pushed to extrapolate and overfit in this direction, increasing the variance of the predictor \cite{Jason2022Memorizing}.
        Hence, these three factors together give rise to fluctuations in the prediction and constitute the variance of the learned predictor function (when averaged over all training datasets).

\section{Proof of Theorem~\ref{thm:testpeak}}
    \label{a:proofthmtestpeak}

    In this section, we provide a formal statement and proof of Theorem~\ref{thm:testpeak}, which was discussed informally in Section~\ref{s:spectralanalysis}.

    \begin{theorem}[Test error peak at interpolation]\label{thm:testpeak}
        Consider a Lipschitz continuous quantum feature map $\rho$ with $p$ linearly independent dimensions, and consider a linear regression problem in the corresponding feature space.
        Let $\{(\rho(x_i),y_i)\}_{i=1}^N$ be a training set, where $x_i \sim \mathcal{N}(0, \mathbb{I}_d)$ are $d$-dimensional i.i.d. Gaussian normal random samples.
        Let $N, p \rightarrow \infty$ where $\frac{p}{N} \rightarrow c \in (0, \infty)$.
        Then, the test error of the quantum linear model based on $\rho(x)$ peaks with high probability at $N=p$.
        
        Furthermore, for a quantum feature map on $n$ qubits, the training set size $N$ at which the peak in test error occurs fulfills $N\in\calO(\exp(n))$ with high probability.
    \end{theorem}

    The following lemma provides the necessary groundwork for the proof of the theorem.

    \begin{lemma}[Generalized Mar\v{c}enko-Pastur law~\cite{louart2018concentration}]\label{lemma:genMP}
        Let $\phi$ be an $L$-Lipschitz continuous map.
        Let $X=[\phi(x_1), \phi(x_2),...,\phi(x_N)]$ be an $(N\times p)$-dimensional matrix with i.i.d. sampled inputs $x_i \sim \mathcal{N}(0, \mathbb{I}_d)$ for $i=1,...,N$.
        For $N, p \rightarrow \infty$ and $\frac{p}{N} \rightarrow c \in (0,  \infty)$, the
        empirical spectral density $\mu_p = \frac{1}{p}\sum_{i}^p\delta_{\lambda_i}$ of the sample covariance matrix $X^\dagger X$ with eigenvalues $\lambda_i$ 
        converges to $\mu_c$ weakly.
        On $(0, \infty)$, $\mu_c$ has continuous density $f_c$ given by
        \begin{align}
             f_c = \begin{cases} \frac{1}{2\pi x c } \sqrt{(x - \lambda_{-})(\lambda_{+} - x)} & \text{if } x \in [\lambda_{-}, \lambda_{+}]\\ 0 & \text{otherwise}\end{cases},
        \end{align}
        where $\lambda_{\pm} = (1\pm\sqrt{c})^2$.
    \label{lemma:gmp}
    \end{lemma}

    \begin{proof}[Proof of Theorem~\ref{thm:testpeak}]
        Let the $L$-Lipschitz continuous feature map on $n$ qubits be given as $\rho = \sum_j \nu_j |\Psi_j\rangle \langle \Psi_j|\in \mathbb{C}^{2^n \times 2^n}$, where $\Psi_j$ is some pure state, $\nu_j$ is the probability of occupying this state and we omit the $x$-dependence for simplicity.
        Recall from Section~\ref{s:preliminaries} that the data matrix $D$ is represented as a vector of density matrices.
        We now introduce an equivalent representation, the vectorized notation of a density matrix, as it aligns more closely with the formulation of the generalized Mar\v{c}enko-Pastur law in Lemma~\ref{lemma:gmp}:
        \begin{align}\label{eq:vectorizednotation}
             | \rho  \rangle\!\rangle &= \sum_j \nu_j \, |\Psi_j\rangle \otimes |\Psi_j^*\rangle \in \mathbb{C}^{4^n}.
        \end{align}
        
        The data matrix can then be written as:
        \begin{align}
            D =  \begin{pmatrix}  \langle\!\langle \rho_1 |\\  \langle\!\langle \rho_2 | \\ \vdots \\   \langle\!\langle\rho_N | \end{pmatrix} \in \mathbb{C}^{N \times 4^n}.
        \end{align}

        Assuming that $N$ and $p$ grow at a constant rate, we can apply Lemma~\ref{lemma:gmp} to the vectorized data matrix $D$.
        Hence, the empirical spectral density of the sample covariance matrix $D^\dagger D$ converges to $\mu_c$ weakly.
        For $\frac{p}{N} \rightarrow 1$, the smallest non-zero eigenvalue of the sample covariance matrix approaches zero, i.e., $\lambda_{-} \rightarrow 0$.
        Note that the eigenvalues of the sample covariance matrix $\lambda_i$ correspond to the squared singular values $\sigma_i$ of the data matrix.
        The error decomposition in Eq.~\eqref{eq:dd_error} revealed the dependence of the test error of quantum linear models on the reciprocal singular values of the data matrix $D$ and hence, the test error is maximized w.h.p. at $N=p$ if $N, p \rightarrow \infty$.

        Additionally, the number of linearly independent dimensions of any $n$-qubit quantum feature map cannot exceed the linear dimension of the vectorized data matrix, which is $4^n$.
    \end{proof}

\section{Proof of Theorem~\ref{thm:lipcontofrhomultidim}}
    \label{a:proofthmlipcont}
    Here, we restate and prove Theorem~\ref{thm:lipcontofrhomultidim}.
    \begin{customtheorem}{2}
        [Lipschitz continuity]
        Let $\rho(x)=S(x)\rho_0S^\dagger(x)$, and let $S(x) = \prod_{l=1}^L\prod_{k=1}^d S_{lk}(x_k)$ (where the product is ordered with increasing $k$ before increasing $l$).
        Let further $S_{lk}(\alpha)=\exp\left(-\frac{i}{2} \alpha H_{lk}\right)$.
        Let $\lambda = \max_{k,l}\lVert H_{lk}\rVert_{\calH}$.
        Then, for any $x,x'\in\bbR^d$, the map $x\mapsto\rho(x)$ is Lipschitz continuous with respect to the operator norm with a Lipschitz constant upper bounded by $\sqrt{dL}\lambda$:
        \begin{align}
            \lVert\rho(x)-\rho(x')\rVert_{\calH} &\leq \sqrt{dL}\lambda \lVert x-x'\rVert.
        \end{align}
    \end{customtheorem}
    
    \begin{proof}
        We start by writing out the difference 
        \begin{align}
            \rho(x) - \rho(x') &= S(x)\rho_0 S^\dagger(x) - S(x' )\rho_0 S^\dagger(x')\\
            & = S(x)\rho_0 S^\dagger(x) - S(x)\rho_0 S^\dagger(x') + S(x)\rho_0 S^\dagger(x') - S(x')\rho_0 S^\dagger(x'),
        \end{align}
        and then use it explicitly to bound the norm of the difference:
        \begin{align}
            \norm{\rho(x) - \rho(x')}_\mathcal{H} &\leq \norm{S(x)\rho_0 S^\dagger(x) - S(x)\rho_0 S^\dagger(x')}_\mathcal{H} + \norm{S(x)\rho_0 S^\dagger(x') - S(x' )\rho_0 S^\dagger(x')}_\mathcal{H}\\
            &= \norm{S(x)\rho_0\, \big( S^\dagger(x) - S^\dagger(x') \big)}_\mathcal{H} + \norm{\big( S(x) - S(x') \big)\, \rho_0 S^\dagger(x')}_\mathcal{H}\\
            &\leq \norm{S(x)}_\mathcal{H}\,\norm{\rho_0}_\mathcal{H}\,\norm{S^\dagger(x) - S^\dagger(x')}_\mathcal{H} + \norm{S(x) - S(x')}_\mathcal{H}\,\norm{\rho_0}_\mathcal{H}\,\norm{S^\dagger(x')}_\mathcal{H}\\
            &\leq \norm{S^\dagger(x) - S^\dagger(x')}_\mathcal{H} + \norm{S(x) - S(x')}_\mathcal{H}.
            \label{eq:Lipsch}
        \end{align}
        
        where in the first equation we add zero and then apply the triangle inequality.
        We further use the sub-multiplicativity of the operator norm, i.e., $ \norm{AB} \leq \norm{A}\norm{B}$ and the fact that unitary matrices fulfill $\norm{S} = 1$ and density matrices $\norm{\rho}\leq 1$.
        We then note that, because each $S_{lk}$ is generated by a single parameter, it holds that $\partial_\alpha S_{lk}(\alpha) = -\frac{i}{2} H_{lk} S_{lk}(\alpha)$.
        Next, we recall that the total derivative of $S(x)$ is a vector of all partial derivatives
        \begin{align}
            \frac{\mathrm d S(x)}{\mathrm dx} = \big( \partial_k S(x) \big)_{k=1}^d.
        \end{align}
        Indeed, the total derivative takes the form of a $3$-tensor, as a vector of matrices.
        Using the mean value theorem~\cite{flett1958mvt}:
        \begin{align}
            \norm{S(x_1) - S(x_2)}_\mathcal{H} \leq \norm{\frac{\text{d}S(x)}{\text{d}x}\Big|_{z\in [x_1, x_2]}}_\mathcal{H} \norm{x_1 - x_2}_2,
        \end{align}
        we need to find an upper bound to the operator norm.
        We invoke Theorem 3.1 in Ref.~\cite{li2016bounds}:
        \begin{align}
            \left\lVert\frac{\mathrm dS(x)}{\mathrm dx}\right\rVert_\calH &\leq \left\lVert \big( \left\lVert \partial_kS(x)\right\rVert_\calH\big)_{k=1}^d\right\rVert_2 \\
            &= \sqrt{\sum_{k=1}^d \left\lVert\partial_kS(x)\right\rVert_\calH^2}\\
            &= \sqrt{\sum_{k=1}^d \left\lVert\partial_k\prod_{l=1}^L\prod_{k'=1}^d S_{l k'}(x_{k'})\right\rVert_\calH^2}
            \label{eq:normsfortensors}
        \end{align}

        To keep the notation clear when applying the chain rule to the partial derivative of the product, we proceed by noting:
        \begin{align}
            S(x) &= S_{11}(x_1) \cdots S_{1d}(x_d) S_{21}(x_1) \cdots S_{Ld}(x_d) \\
            \partial_kS(x) &= S_{11}(x_1) \cdots \partial_kS_{1k}(x_k) \cdots S_{Ld}(x_d) + \cdots + S_{11}(x_1) \cdots \partial_kS_{Lk}(x_k)\cdots S_{Ld}(x_d) \\
            &= -\frac{i}{2}\left(S_{11}(x_1) \cdots H_{1k}S_{1k}(x_k) \cdots S_{Ld}(x_d)+ \cdots + S_{11}(x_1) \cdots H_{Lk}S_{Lk}(x_k)\cdots S_{Ld}(x_d)\right)
        \end{align}
        
        Inserting this into Eq.~\eqref{eq:normsfortensors} and using the fact that matrix norms are invariant under multiplication with unitary matrices yields
        \begin{align}
            \left\lVert\frac{\mathrm dS(x)}{\mathrm dx}\right\rVert_\calH &\leq \sqrt{\sum_{k=1}^d \left\lVert\partial_k\prod_{l=1}^L\prod_{k'=1}^d S_{l k'}(x_{k'})\right\rVert_\calH^2} \\
            &= \sqrt{\sum_{k=1}^d \left\lVert -\frac{i}{2}\left(S_{11}(x_1) \cdots H_{1k}S_{1k}(x_k) \cdots S_{Ld}(x_d)+ \cdots + S_{11}(x_1) \cdots H_{Lk}S_{Lk}(x_k)\cdots S_{Ld}(x_d)\right)\right\rVert_\calH^2} \\
            &\leq \frac{1}{2}\sqrt{\sum_{k=1}^d \left\lVert S_{11}(x_1) \cdots H_{1k}S_{1k}(x_k) \cdots S_{Ld}(x_d)\right\rVert_\calH^2 + \cdots +\left\lVert S_{11}(x_1) \cdots H_{Lk}S_{Lk}(x_k)\cdots S_{Ld}(x_d)\right\rVert_\calH^2} \\
            &= \frac{1}{2}\sqrt{\sum_{k=1}^d \left\lVert H_{1k}\right\rVert_\calH^2 + \cdots +\left\lVert H_{Lk}\right\rVert_\calH^2} \\
            &\leq \frac{1}{2} \sqrt{\sum_{k=1}^d L\lambda^2} \\
            &= \frac{\sqrt{dL}\lambda}{2}.
        \end{align}
    
        Plugging this into Eq.~\eqref{eq:Lipsch} above, we obtain:
        \begin{align}
            \lVert \rho(x)-\rho(x')\rVert_\calH &\leq 2 \left\lVert\frac{\mathrm dS(x)}{\mathrm dx}\right\rVert_\calH\lVert x-x'\rVert \\ 
            &\leq 2 \frac{\sqrt{dL}\lambda}{2} \lVert x-x'\rVert \\ 
            &= \sqrt{dL}\lambda \lVert x-x'\rVert,
        \end{align}
        which completes the proof.  
    \end{proof}

    We highlight that a similar result was recently reported in Ref.~\cite{recio2024single} for the $1$-norm.

\section{Additional experiments and details}\label{a:numerics}

Here, we provide further implementation details of the numerical experiments, as well as additional numerical results. In all experiments, we consider QML models based on embedding or projected quantum kernels, evaluated on synthetic and real-world datasets. The number of test samples is fixed to $100$, while the number of training samples is varied to explore the transition across underparameterized and overparameterized regimes.

The structure of the quantum circuits used to define the feature maps consists of layers of single-qubit rotations, followed by CNOT gates. This architecture is repeated to ensure the kernel matrix is full rank. The kernel rank is upper bounded by $\min(N, p)$, where $N$ is the number of training samples and $p$ is the feature space dimension. In the underparameterized regime ($N > p$), expressivity limits the rank to at most $p$, while in the overparameterized regime ($N < p$), sufficient expressivity ensures the kernel reaches full rank $N$.
Classical inputs are preprocessed to match the dimensionality constraints of the quantum model. In particular, dimensionality reduction is applied to ensure that inputs can be embedded into circuits with limited number of qubits. For this purpose, principal component analysis is used to project the data to a space of dimension $2n$, where $n$ is the number of qubits used in the corresponding feature map.

\paragraph{Ablation experiments:}

We extend the ablation experiments described in the main text to real-world datasets, see Fig.~\ref{fig:ablation_realdata}. These include the Fashion MNIST and California Housing datasets. In one set of experiments, we control the spectral properties of the training data by applying a cutoff to the singular values of the kernel matrix. Removing low-magnitude components mitigates the peak in test error, confirming that small singular values play a key role in amplifying generalization error near interpolation. In a second set of experiments, we assess the impact of feature alignment between training and test data by projecting test samples onto subsets of the leading singular directions derived from the training set. As the number of retained directions decreases, the test data becomes less aligned with the training data, leading to a diminished peak. Lastly, we also perform the residual error ablation experiments. The results on real-world data qualitatively mirror those from the synthetic case, showing that eliminating the residual error significantly reduces or removes the double descent peak. This further supports the theoretical decomposition and confirms that residual error plays a central role in the emergence of the test error peak.

\paragraph{Training error:}
In addition to the test error curves shown in Fig.~\ref{fig:dd_qkernel}, we present the corresponding training error curves in Fig.~\ref{fig:mse_train}. As expected, the training error remains close to zero in the overparameterized regime ($N/4^n < 1$), reflecting the model's capacity to fit the training data almost perfectly. In contrast, in the underparameterized regime ($N/4^n > 1$), the training error is small but non-zero, indicating that the model cannot fully capture the training data due to limited capacity. This behavior aligns well with theoretical expectations and clearly illustrates that the double descent phenomenon is not rooted in the training error behavior.

\paragraph{Sampling noise:} 

To analyze how sampling noise affects the double descent curve, we performed a numerical study with varying numbers of shots on $n=2$ qubits and the Fashion MNIST dataset, see Fig.~\ref{fig:noise}. For the simulation of $n_\text{shots}$ shots we add Gaussian noise with variance $\sigma^2  = K_{ij}(K_{ij}-1)/n_\text{shots}$ to each entry $K_{ij}$ of the Gram matrix as well as on kernel evaluations on test points. The perturbed Gram matrix is then projected back to a positive-semi definite matrix by setting all negative eigenvalues to zero and reconstructing the matrix from its eigen-decomposition. Effectively, a limited number of shots constrains the precision with which the Gram matrix can be estimated. To ensure consistency, we also limit the numerical precision when computing its spectrum to approximately $\log_{10}(n_\text{shots}/2)$ decimals. Since the test error depends on the inverse of the non-zero eigenvalues of the Gram matrix, shot noise acts similarly to a regularization, effectively smoothing the curve and reducing the prominence of the peak.
%Note that the number of shots required to accurately reconstruct the noise-free curves is prohibitively large and beyond practical feasibility. However, the central focus of our work is not on demonstrating the visibility of the double descent peak under all conditions. Instead, our contribution lies in establishing that quantum kernel methods inherently exhibit the statistical structure necessary for double descent, highlighting their capacity to operate in the overparameterized regime.

\paragraph{Regularization:}

A common practice in kernel methods is to add a regularization term of the form $\lambda \lVert \calM \rVert^2$ to the optimization problem, where the regularization parameter $\lambda \geq 0$ controls the trade-off between accurately fitting the training data and promoting a smoother, more stable solution. Fig.~\ref{fig:regularization} illustrates how increasing $\lambda$ affects the double descent behavior in quantum kernel methods for $n=3$ qubits on the Synthetic dataset. As regularization strength grows, the peak in the test error curve becomes less pronounced and eventually disappears. 

\begin{figure*}[t!]
        \centering
        \includegraphics[width=\textwidth]{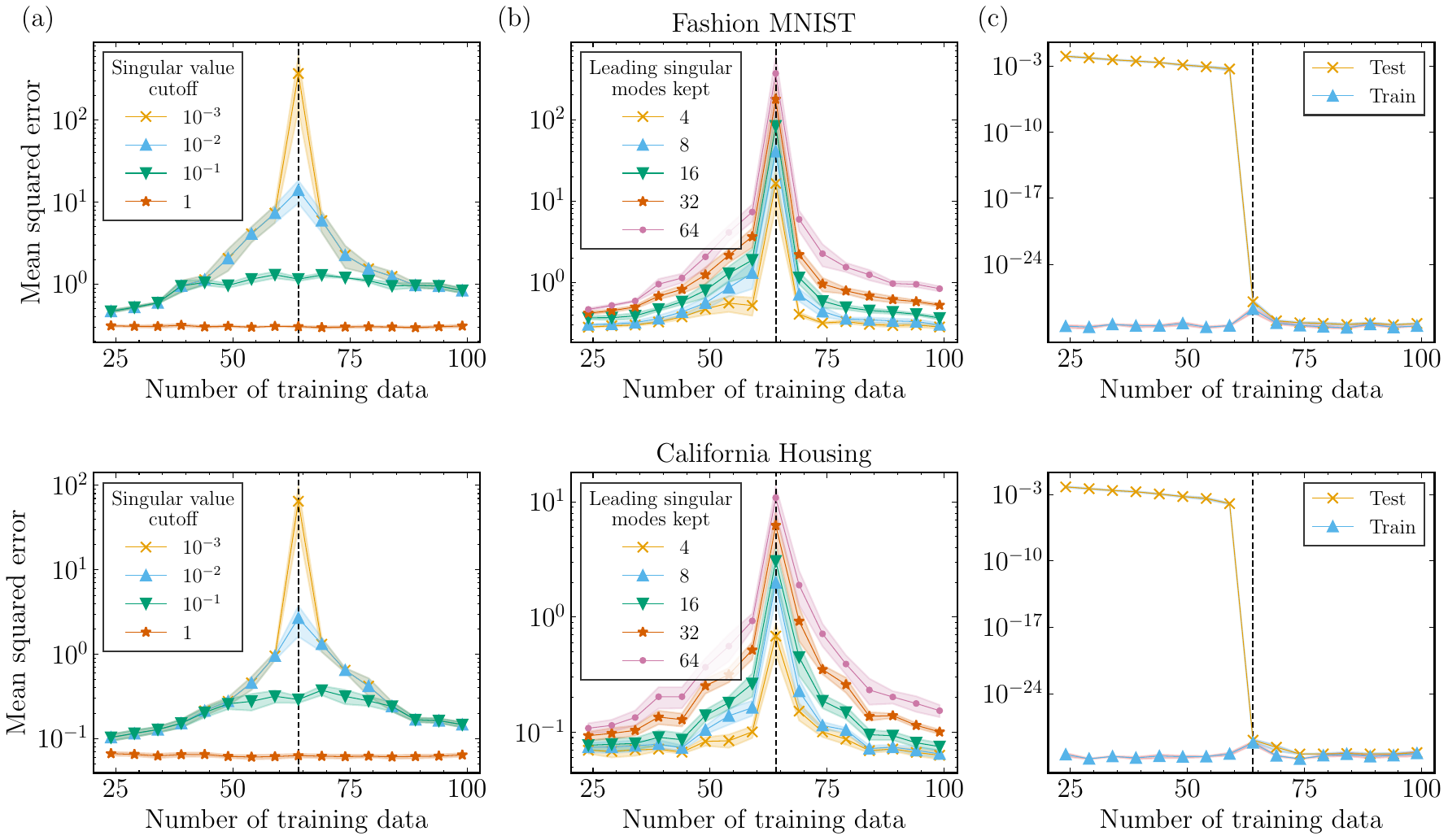}
        \caption{\textbf{Ablation experiments for real-world datasets.} Mean test squared error as a function of the number of training data points $N$ for (a) applying different cutoffs on the minimum singular value, (b) retaining varying numbers of leading singular modes of the input data, and (c) eliminating residual error.
        All experiments were performed on the Fashion MNIST (top) and California Housing (bottom) datasets with $n=3$ qubits.
        The shaded area corresponds to the standard deviation for five independent experiment repetitions, each using independently sampled training data. 
        The dotted black line indicates the interpolation threshold. The region to the left of the threshold corresponds to the overparameterized regime, while the region to the right corresponds to the underparameterized regime.
        Modifying each factor reduces or eliminates the double descent peak, highlighting their roles in the test error behavior.}
        \label{fig:ablation_realdata}
    \end{figure*}

    \begin{figure*}[t!]
        \centering
        \includegraphics[width=\textwidth]{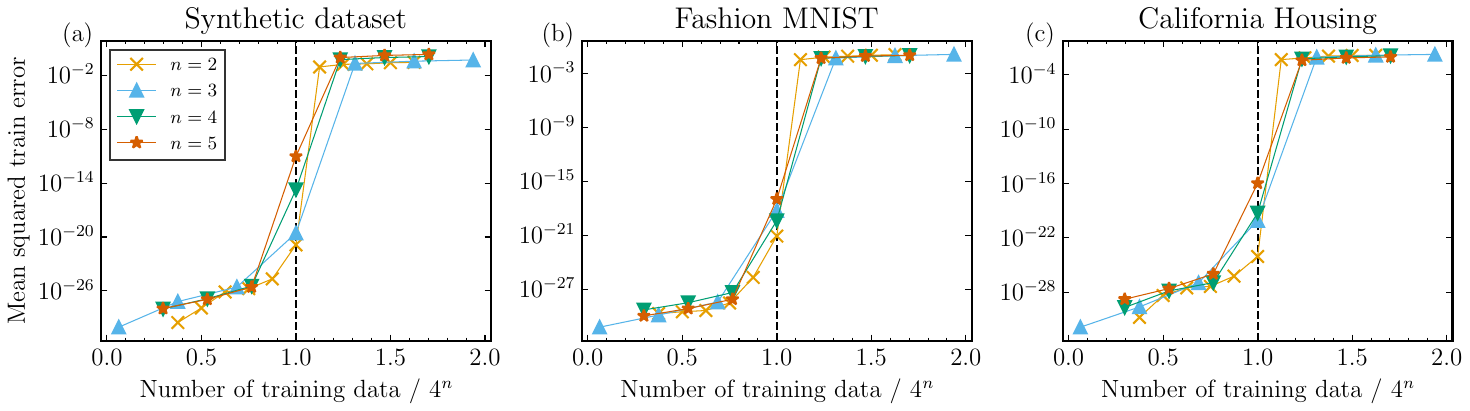}
        \caption{\textbf{Mean squared train error.}
            Mean squared train error as a function of the normalized number of training data points $(N/4^n)$ for the (a) Synthetic, (b) Fashion MNIST, and (c) California Housing datasets, employing the EQK in Eq.~\eqref{eq:eqk}.
            The dotted black line indicates the interpolation threshold. The region to the left of the threshold corresponds to the overparameterized regime, while the region to the right corresponds to the underparameterized regime.}
        \label{fig:mse_train}
    \end{figure*}

    \begin{figure*}[t!]
        \centering
        \includegraphics[width=0.5\textwidth]{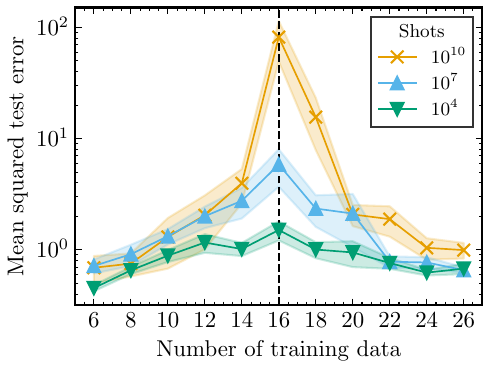}
        \caption{\textbf{Sampling noise.}
             Mean squared test error as a function of the number of training data points $N$ for different numbers of shots. The experiments were performed on the Fashion MNIST dataset with $n=2$ qubits. The dotted black line indicates the interpolation threshold. The region to the left of the threshold corresponds to the overparameterized regime, while the region to the right corresponds to the underparameterized regime.}
        \label{fig:noise}
    \end{figure*}

    \begin{figure*}[t!]
        \centering
        \includegraphics[width=0.5\textwidth]{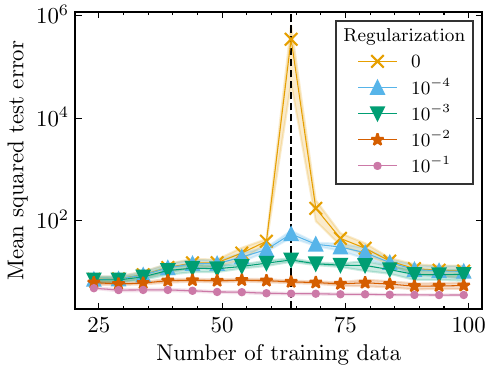}
        \caption{\textbf{Regularization.} Mean squared test error as a function of the number of training data points $N$ for for several values of the regularization parameter $\lambda$. The experiments were performed on the Synthetic dataset with $n=3$ qubits. The dotted black line indicates the interpolation threshold. The region to the left of the threshold corresponds to the overparameterized regime, while the region to the right corresponds to the underparameterized regime.}
        \label{fig:regularization}
    \end{figure*}

\end{document}